\documentclass[journal]{IEEEtran}

\usepackage{amsmath,amssymb, amsthm}
\usepackage{booktabs}
\usepackage{hyperref}
\usepackage{enumitem}
\usepackage{algorithm}
\usepackage{algpseudocode}
\usepackage{mathtools}
\usepackage{graphicx}
\usepackage{xcolor}
\usepackage{tikz}
\usepackage{balance}
\usepackage{float}

\usetikzlibrary{arrows.meta, positioning, fit, backgrounds, shapes.geometric, patterns}

\newtheorem{theorem}{Theorem}
\newtheorem{lemma}{Lemma}
\newtheorem{proposition}{Proposition}
\newtheorem{corollary}{Corollary}

\theoremstyle{definition}

\theoremstyle{remark}
\newtheorem{remark}{Remark}

\title{Trust-Aware Topology Learning for Dynamic Decentralized Federated Learning under Adversaries}
\author{Shubham~Vaishnav$^*$, Murtaza~Rangwala$^*$, Ali~Beikmohammadi, Sindri~Magnússon, and~Rajkumar~Buyya
\thanks{$^*$These authors contributed equally to this work.}
\thanks{M. Rangwala and R.Buyya are with the Quantum Cloud Computing and Distributed Systems (qCLOUDS) Laboratory at the University of Melbourne, Australia (e-mail: murtaza.rangwala@student.unimelb.edu.au; rbuyya@unimelb.edu.au).}
\thanks{S. Vaishnav, A. Beikmohammadi, and S. Magn\'usson are with the Department of Computer and System Science, Stockholm University, Sweden (e-mail: shubham.vaishnav@dsv.su.se; beikmohammadi@dsv.su.se; sindri.magnusson@dsv.su.se).}
}

\begin{document}
\maketitle

\begin{abstract}
% Decentralized federated learning (DFL) lets edge devices train personalized models without a central server, but it exposes every device to its peers: in dynamic, mobile deployments an adversary can poison both the model updates and the \emph{topology information} that devices use to decide whom to learn from.
In dynamic mobile decentralized federated learning (DFL), where
device connectivity varies over time with movement, failures, and
adversarial behavior, an adversary can poison both the model
updates and the \emph{topology information} that devices use to decide
whom to learn from.
We present DMTT (Dynamic \texttt{MURMURA} with Trusted Topology), a decentralized personalized federated learning protocol built on
\texttt{MURMURA}, a prior framework that uses evidential deep learning to assess
peer model compatibility and down-weight distribution-mismatched collaborators.
DMTT extends this foundation to time-varying graphs under topology-manipulation
attacks.
Each device maintains a confidence-weighted local topology view from direct
link-reliability estimates, signed topology claims, witness corroboration, and a
Beta-distributed source-trust model, and then aggregates only over a
trust-screened collaborator set using a composite score that fuses model
compatibility, topology trust, and link reliability. We prove that the induced
screened mixing matrices confine Byzantine influence to a bounded residual
$\delta_{\max}$ that vanishes under perfect screening, and we implement DMTT as a
coordinator-free distributed framework in which each client runs as an independent
ZeroMQ process synchronised by a shared wall-clock epoch. On UCI HAR and PAMAP2, each
partitioned across $100$ mobile clients with Dirichlet heterogeneity, DMTT sustains honest-node accuracy above $0.862$ at every tested fraction from
$10\%$ to $80\%$ on UCI HAR and above $0.829$ on PAMAP2, essentially matching no-attack accuracy at low
fractions and degrading gracefully to near-local-only performance at extreme fractions. Static
and dynamic FedAvg collapse to chance at every fraction, and dedicated robust
aggregators (Krum, BALANCE, UBAR) all fail to consistently beat a safe local-only
baseline that forgoes collaboration entirely. DMTT is the only method that clears
the local-only bar across both datasets at all tested fractions, demonstrating that
trust-aware screening both resists poisoning and profits from honest peers.
Empirically the surviving Byzantine aggregation weight is exactly zero across all
runs, consistent with the $\delta_{\max}=0$ condition of the theoretical analysis. The same
protocol runs end-to-end on a coordinator-free distributed ZeroMQ backend that
reproduces these results on real nodes from the Melbourne Research Cloud.
\end{abstract}


\section{Introduction}
Federated learning (FL) lets many devices train a shared model together without ever sending their raw data to a central server. It is now used in mobile keyboards, healthcare, and industrial IoT. A core difficulty in any realistic FL deployment is that no two devices
see the same kind of data: a smartwatch on one wearer records different activity patterns from one on another. A hospital sees a
different patient mix from its neighbour, and so on. With this kind
of \emph{non-IID} (non-identically-distributed) data, a single global
model rarely fits every device well. Each device benefits from a
\emph{personalized} model that reflects its own data, while still
learning from peers whose data is related but not identical. Standard
FL still relies on one central server to coordinate this, which
becomes both a privacy bottleneck and a single point of failure.
Decentralized federated learning (DFL) removes that server entirely:
devices exchange and aggregate model updates directly with
peers~\cite{ye2022decentralized,kharrat2024decentralized,liu2024decentralized}.
The system becomes more private and more resilient, but every device
is now exposed to its peers. Without a central authority, each
device must decide on its own which peers to learn from, even though
any participant can in principle behave adversarially.

\texttt{MURMURA}~\cite{rangwala2025evidential} addresses this
per-device decision problem under non-IID data. Its central
observation is that not every peer is a useful collaborator. A peer
trained on a very different distribution can actively hurt local
performance, even if the peer is honest and well-trained.
\texttt{MURMURA} calls this property \emph{peer compatibility} and
quantifies it from each device's local point of view. The device runs
an incoming peer model on its own data and looks at \emph{how} the
peer model is uncertain about that data. Using evidential deep
learning, it separates two kinds of uncertainty in the peer's
predictions: uncertainty because the peer has never really seen data
like this (a sign of incompatibility), and uncertainty because the
data is just inherently tricky (which is not the peer's fault).
Peers showing the first kind are down-weighted during aggregation. Peers showing the second kind are kept. Each round produces a fresh
per-device compatibility score, so the set of useful collaborators
adapts as training progresses.

However, \texttt{MURMURA} does not solve the systems-level question of how peers discover and maintain knowledge of a changing network topology in the first place. This gap becomes particularly sharp in dynamic DFL, where connectivity varies over time due to mobility, failures, or adversarial behavior. Existing \texttt{MURMURA} evaluations assume fixed graph structures and synchronous rounds. In contrast, dynamic topology works such as Soft-DSGD and ToLRDUL either rely on link reliability models that do not consider malicious topology reports or assume globally coordinated topology updates that are too strong for fully decentralized environments \cite{rangwala2025evidential,ye2022decentralized,wu2024topology}. Therefore, there is a need for dynamic \texttt{MURMURA} with adversary-aware topology trust.

Beyond practicality, there are also strong theoretical reasons why this gap has persisted. Nesterenko and Tixeuil show that under Byzantine faults, weak topology discovery requires at least $k+1$-connectivity, while strong topology discovery requires more than $2k+1$-connectivity to tolerate $k$ Byzantine nodes \cite{nesterenko2009discovering}.

These requirements make exact global topology agreement impractical for sparse wearable, IoT, or mobile FL deployments. As a result, the appropriate target is not exact topology discovery, but rather a confidence-weighted topology view: exact for direct neighbors, reasonably reliable for two-hop structure, and probabilistic beyond that.

This work proposes a decentralized personalized FL framework that extends \texttt{MURMURA} beyond a fixed-neighborhood setting. It maintains a confidence-weighted topology state built from direct observation, signed claims, corroboration, and uncertainty-aware source trust. It then uses this trusted graph to choose collaborators using evidential trust. The proposed system can be viewed as \texttt{MURMURA} augmented with a trustworthy topology memory. The design separates the problem into three distinct layers as shown in Table~\ref{tab:three_layers}.

\begin{table}[t]
\centering
\caption{Three-layer separation between reachability, topology trust, and collaboration}
\label{tab:three_layers}
\begin{tabular}{@{}p{0.28\columnwidth} p{0.66\columnwidth}@{}}
\toprule
\textbf{Layer} & \textbf{Description} \\
\midrule

\textbf{Direct reachability} &
Tracks which peers can be contacted in the current epoch based on local observations such as authenticated beacons and packet acknowledgments~\cite{ye2022decentralized}. \\[0.6ex]

\textbf{Topology trust} &
Maintains confidence in topology claims using signed reports, corroboration across multiple witnesses, historical consistency, and path diversity~\cite{nesterenko2009discovering,gaucher2024unified}. \\[0.6ex]

\textbf{Collaboration graph} &
Determines which reachable and trustworthy peers should influence local training, combining \texttt{MURMURA} model trust, topology trust, and communication cost~\cite{rangwala2025evidential,kharrat2024decentralized,liu2024decentralized}. \\

\bottomrule
\end{tabular}
\end{table}

\paragraph{Contributions.}
\begin{enumerate}[leftmargin=2em]
  \item \textbf{DMTT protocol with convergence guarantees.} A four-algorithm protocol integrating link-reliability tracking, witness-weighted signed topology claims, Beta-distributed source trust, and trust-aware collaborator selection. We prove that the induced mixing matrices confine Byzantine influence to a bounded residual that vanishes under perfect screening.
  \item \textbf{Fully distributed \texttt{MURMURA} framework.} \texttt{MURMURA} is redesigned from single-process simulation into independent ZeroMQ processes synchronised by a shared wall-clock epoch, with a passive PULL-only monitor collecting metrics without influencing training.
  \item \textbf{Mobility-driven dynamic topology and topology-liar attack.} A deterministic random-walk mobility model that generates $\mathcal{G}^t$ identically across all nodes from a shared seed, plus an oracle-free self-inclusion check: if $j$ delivered a message to $i$ this round, $j$'s topology claim must include $i$. We introduce an \emph{omission} topology-liar attack in which Byzantine nodes claim only the Byzantine coalition, so the self-inclusion contradiction fires every round.
  \item \textbf{Empirical validation.} On UCI HAR and PAMAP2 under combined topology-liar and Gaussian poisoning at Byzantine fractions of $10\%$--$80\%$, DMTT attains the highest honest-node accuracy in every setting while non-robust baselines collapse to chance; a trust diagnostic confirms zero aggregation weight to Byzantine peers, instantiating the perfect-screening regime of the theory.
\end{enumerate}

\section{Related Work}
\label{sec:related}

\paragraph{Decentralized and personalized federated learning.}
Decentralized federated learning removes the central server and allows clients to exchange model updates over a peer-to-peer graph. Decentralized FedAvg and related local-update methods show that convergence and communication cost depend on the number of local updates, the frequency of communication, and the mixing matrix~\cite{9850408,9713700}. This design is useful for edge and mobile systems because it avoids a server bottleneck and a single point of failure.

However, clients usually have non-IID data, so a fixed neighborhood is not always useful. Some decentralized personalized-learning methods therefore learn collaboration graphs or adaptive mixing weights from the relationships between clients and their local tasks~\cite{pmlr-v108-zantedeschi20a,Li_2022_CVPR,ye2023personalized}. Directed collaboration also allows clients with different data, computing, and communication conditions to use different peer sets~\cite{liu2024decentralized}. MURMURA~\cite{rangwala2025evidential} follows a related model-level approach. Each client tests an incoming model on its own data and uses evidential uncertainty to identify peers whose data is too different from its own~\cite{sensoy2018evidential}. DMTT keeps this idea of client-specific model compatibility, but also considers whether the topology information used to select peers can be trusted.

\paragraph{Dynamic topologies and mobile decentralized learning.}
Several studies examine decentralized learning when links are unreliable or the graph changes over time. Soft-DSGD uses link-dependent mixing weights to handle packet loss and other communication problems~\cite{ye2022decentralized}. ToLRDUL selects links by considering data differences, representation differences, and unreliable D2D connections~\cite{wu2024topology}. Koloskova et al. analyze local updates and changing mixing matrices in a common framework~\cite{koloskova2020unified}. Other studies examine optimization over directed, time-varying graphs and show the importance of connectivity over time~\cite{nedic2014distributed,7405263}. For non-IID data, the choice of neighbors matters as well as the number of neighbors. Neighborhood heterogeneity can affect the convergence bound, and a carefully learned sparse graph can reduce this problem by connecting clients with complementary data~\cite{le2023refined}.

Recent work has applied these ideas to mobile and asynchronous DFL. Lu et al. use the Metropolis--Hastings rule to update mixing weights when the communication graph changes. They also use secret sharing to protect model states. However, their threat model assumes semi-honest participants and does not cover active attacks~\cite{lu2023privacy}. DySTop controls update staleness and builds a changing topology for asynchronous DFL. Its topology construction considers non-IID data, delay, and communication cost~\cite{shi2026dystop}. Li et al. study DFL over time-varying and heterogeneous mobile computing networks. Their analysis links neighborhood heterogeneity with the properties of the communication graph~\cite{11271539}. These studies are close to DMTT because they consider changing mobile graphs. However, they assume that topology and link information are honest. They do not consider a Byzantine client that hides honest neighbors, invents links, or sends false topology reports.

\paragraph{Byzantine resilience and trusted topology.}
Another line of work studies Byzantine-resilient decentralized learning. Most of these methods assume that the communication graph is already known and focus on filtering model or gradient messages. BRIDGE studies Byzantine failures in directed peer-to-peer networks and gives convergence guarantees under graph connectivity conditions~\cite{9815556}. UBAR provides a Byzantine-resilient aggregation rule for decentralized learning. Iterative outlier scissor also shows that local robust aggregation can create disagreement and mixing errors if the resulting virtual mixing process is not properly controlled~\cite{guo2021byzantine,wu2023byzantine}. Self-centered clipping shows that topology bottlenecks can make Byzantine attacks more harmful. It gives a convergence bound that depends on the remaining Byzantine influence and the graph structure~\cite{he2022byzantine}. BALANCE uses each client's local model as a reference for filtering poisoned peer models and provides convergence results for decentralized FL under poisoning attacks~\cite{fang2024byzantine}. CMFL uses rotating committees to score local gradients and select updates in a serverless FL system~\cite{9870745}. Other work studies Byzantine resilience under heavy-tailed data and compressed communication, but it still uses a central server for aggregation~\cite{10070815}.

These model-level defenses usually start after a peer has already entered the communication or aggregation neighborhood. A Byzantine client may therefore lie about reachability, adjacency, or witness information before its model update is filtered. Work on Byzantine topology discovery shows that exact global topology recovery requires strong connectivity assumptions. Weak discovery needs connectivity greater than the number of Byzantine faults, while strong discovery needs connectivity greater than twice that number~\cite{nesterenko2009discovering}. This makes a confidence-weighted local topology view more realistic than a claim of perfect global knowledge. In low-power IoT routing, TRAIL uses reachability tests to detect topological attackers without relying on heavy cryptography~\cite{perrey2013trail}. In mobile ad hoc networks, reputation systems combine direct observations with recommendations from other nodes while reducing the effect of stale or inconsistent information~\cite{7345272}. Blockchain-based FL can provide a stronger audit trail by combining malicious-gradient detection, homomorphic protection, and blockchain records, but it also adds infrastructure and coordination costs~\cite{9849010}.

DMTT combines these ideas through three separate decisions. Direct reachability determines which peers can exchange messages at the current time. Topology trust determines whether a peer's claims about the wider graph should affect peer selection. MURMURA model compatibility determines whether that peer is useful for local personalization. Dynamic-topology methods usually assume that the graph information is honest, while Byzantine aggregators usually filter only model updates. DMTT addresses both problems by screening topology claims before building the collaborator set and then screening the models of the remaining peers. This produces a screened mixing matrix whose Byzantine influence can be bounded while keeping the implementation fully distributed.

\section{Problem Formulation}
\label{sec:problem}

This section states the system, data, and adversary model, identifies the local state each client must maintain, and writes down the optimization problem that the protocol of Section~\ref{sec:protocol} solves. Concrete update rules and equations for each state variable are deferred to Section~\ref{sec:protocol}.

\subsection{System and Communication Model}

We consider a decentralized personalized FL system with a fixed set of clients $\mathcal{V} = \{1,2,\dots,N\}$.
Client identities are persistent, but client positions vary over time. Let ${\bf r}_i^t$ denote the position of client $i$ at epoch $t$. Mobility induces a time-varying communication topology
$
\mathcal G^t = (\mathcal V,\mathcal E^t),
$ where $(i,j)\in\mathcal E^t$ if clients $i$ and $j$ are physically able to communicate at epoch $t$. The node set is fixed; the edge set changes due to motion, channel variation, or obstruction. This is the dynamic-topology extension beyond the fixed-neighborhood setting of \texttt{MURMURA}~\cite{rangwala2025evidential}, and it follows the standard time-varying graph model used in decentralized optimization~\cite{ye2022decentralized,wu2024topology}.

\subsection{Local Data and Personalized Objective}

Each client $i\in\mathcal V$ holds a private dataset $\mathcal D_i = \mathcal D_i^{tr} \cup \mathcal D_i^{val}$ drawn from a client-specific distribution $P_i$. The distributions $\{P_i\}$ are generally non-IID, so the system seeks not a single global model but a collection of personalized models $\Theta = \{\theta_1,\dots,\theta_N\}$ with $\theta_i\in\mathbb R^d$, following the personalization principle of \texttt{MURMURA}~\cite{rangwala2025evidential}. The local population risk at client $i$ is
$$
F_i(\theta_i) = \mathbb E_{(x,y)\sim P_i}\big[\ell(f(x;\theta_i),y)\big],
$$
and the population objective is
$$
\min_{\Theta}\; \sum_{i=1}^{N} F_i(\theta_i).
$$
Unlike conventional DFL, client $i$ cannot safely collaborate with every other client at every epoch: communication feasibility and topology trust are both time-varying.

\subsection{Adversary Model}
\label{sec:adversary}

Let $\mathcal H \subseteq \mathcal V$ denote the honest clients and $\mathcal B = \mathcal V \setminus \mathcal H$ the Byzantine clients. The Byzantine set is fixed over the training horizon, although Byzantine clients move like honest ones. Honest clients follow the prescribed protocol. Byzantine clients may behave arbitrarily, including issuing false edge-addition or edge-removal claims, replaying stale topology states, flooding the network with weak claims, or sending misleading model updates. The formulation therefore separates topology trust from model usefulness so that topology misinformation is filtered before collaboration decisions are made.

\subsection{Local State to Be Maintained}
\label{sec:local_state}

At every epoch $t$, each client $i$ must maintain four categories of local state, whose precise update rules are given in Section~\ref{sec:protocol}:
\begin{itemize}[leftmargin=1.4em,itemsep=0.2ex]
    \item \emph{Direct reachability and link reliability.} A one-hop neighbor set $N_i^{dir,t}$ from authenticated hello beacons and acknowledgments, and a scalar link-reliability estimate $\hat p_{ij}^t \in [0,1]$ tracking communication quality (\emph{not} topology trust) \cite{ye2022decentralized}.
    \item \emph{Signed topology claims and witness evidence.} Each locally observed topology event is emitted as a signed, versioned claim; received claims are collected in per-edge support and opposition buffers $\mathcal W_i^{+,t}(e), \mathcal W_i^{-,t}(e)$, which the protocol summarizes into an edge-confidence score $\chi_i^t(e)\in[0,1]$.
    \item \emph{Source-trust state.} For each peer $j$, a Beta-style evidence state $(\alpha_{ij}^t,\beta_{ij}^t)$ tracks the long-run reliability of $j$ as a topology reporter, from which the protocol derives a reputation $R_{ij}^t$, an epistemic uncertainty $U_{ij}^t$, and a topology-trust score $T_{ij}^{topo,t}\in[0,1]$ (the topology-side analogue of the model-side evidential trust of \texttt{MURMURA}~\cite{rangwala2025evidential}).
    \item \emph{Trusted local topology view.} By accepting an edge $e$ only when $\chi_i^t(e)$ and corroboration are sufficient, client $i$ maintains a confidence-weighted local view $\hat{\mathcal G}_i^t = (\mathcal V,\hat{\mathcal E}_i^t)$. The target is not exact global topology discovery (which would require $(2k{+}1)$-connectivity under $k$ Byzantine faults~\cite{nesterenko2009discovering}) but a locally sufficient view for safe collaboration.
\end{itemize}

\subsection{Trusted Feasible Set and Effective Learning Graph}

Given the local state above, the protocol must induce, at every epoch, a \emph{trusted feasible neighbor set} $\mathcal F_i^t \subseteq N_i^{dir,t}$ restricted to peers that are simultaneously reachable (sufficient $\hat p_{ij}^t$) and trustworthy as topology reporters (sufficient $T_{ij}^{\mathrm{topo},t}$), and an \emph{active collaborator set} $C_i^t \subseteq \mathcal F_i^t$ satisfying $|C_i^t| \le B$ within an edge budget $B$, as in budgeted decentralized personalized graph learning~\cite{kharrat2024decentralized}. The corresponding aggregation weights $\{w_{ij}^t\}$ define a row-stochastic, time-varying \emph{screened mixing matrix} $W^t = [w_{ij}^t]_{i,j=1}^N$ with
$$
w_{ij}^t = 0 \;\text{if}\; j\notin C_i^t\cup\{i\},\quad
\sum_{j} w_{ij}^t = 1,\quad w_{ij}^t \ge 0.
$$
Hence the effective learning graph is not the raw physical graph $\mathcal G^t$ but a client-specific, time-varying, trust-screened graph. This connects the formulation to standard analyses of decentralized learning over time-varying graphs, where topology enters through the support and spectral properties of the mixing matrices.

\subsection{Optimization Problem}

The problem addressed in this work is to jointly design four protocol rules:
\begin{enumerate}[leftmargin=2em,itemsep=0.2ex]
  \item the topology-claim processing rule, producing $\hat{\mathcal G}_i^t$;
  \item the source-trust update rule, producing $T_{ij}^{topo,t}$;
  \item the trusted-feasible-neighbor selection rule, producing $\mathcal F_i^t$ and $C_i^t$;
  \item the personalized aggregation rule, producing $W^t$ and $\theta_i^{t+1}$.
\end{enumerate}
Over a horizon of $T$ epochs, the goal is to minimize the time-averaged aggregate personalized risk
$$
\min_{\{\theta_i^t\},\,\Pi_{topo},\,\Pi_{collab}}\;
\frac{1}{T}\sum_{t=1}^{T}\sum_{i=1}^{N} \mathbb E\big[F_i(\theta_i^t)\big],
$$
subject, for every epoch $t$ and client $i$, to
$$
C_i^t \subseteq \mathcal F_i^t \subseteq N_i^{dir,t},\qquad |C_i^t|\le B,
$$
the row-stochastic mixing-matrix constraints above, and the requirement that topology information may influence collaboration only through the screened local topology state $(\hat{\mathcal G}_i^t, T_{ij}^{topo,t})$.

In words, the proposed method, Dynamic \texttt{MURMURA} with Trusted Topology (DMTT), must maintain a trustworthy, confidence-weighted topology view and perform personalized decentralized learning only over the resulting trusted feasible graph. Section~\ref{sec:protocol} specifies the four rules above and gives the equations defining $\hat p_{ij}^t$, $\chi_i^t(e)$, $T_{ij}^{topo,t}$, and $W^t$.

\section{The DMTT Protocol}
\label{sec:protocol}

The proposed method is a decentralized training protocol with an explicit topology-trust layer. Instead of trusting a topology claim on arrival, each node first grounds its network view in direct neighbor observations, then verifies signed claims for freshness and authenticity, accumulates corroborating and opposing witness evidence, updates long-term trust in each reporting source, and limits forwarding through trust thresholds, hop limits, and rate budgets. Only after this topology-filtering stage does the node invoke \texttt{MURMURA}-style model trust to decide which reachable peers should influence training.

The protocol contains six logical steps, but for clarity they are grouped into four algorithmic modules. Algorithm~\ref{alg:main} integrates direct reachability estimation, signed topology-claim generation, local model training, model exchange, and collaborator selection. Topology-claim processing and edge-confidence computation are handled by Algorithm~\ref{alg:claim}. Controlled forwarding is handled by Algorithm~\ref{alg:forward}. Source-trust updating is handled by Algorithm~\ref{alg:trust}.

\textbf{Local State at Node $i$:} Each node $i$ maintains local state comprising direct link quality, topology trust, witness evidence, and model-level collaboration scores. Algorithm~\ref{alg:main} gives the overall decentralized protocol, executed locally and in parallel at every epoch.

\begin{algorithm}[t]
\caption{DMTT: Dynamic \texttt{MURMURA} with Trusted Topology}
\label{alg:main}
\begin{algorithmic}[1]
\State Initialize local model $\theta_i^0$
\State Initialize direct-neighbor set $N_i^{dir}$
\State Initialize trust state $\{\alpha_{ij},\beta_{ij},T_{ij}^{topo}\}_j$
\State Initialize witness buffers $\mathcal{W}_i(e)$ and edge confidences $\chi_i(e)$
\State Initialize collaborator set $C_i$
\For{each epoch $t$ executed in parallel at all nodes}
    \For{each authenticated neighbor $j$ heard in epoch $t$}
        \State Add $j$ to $N_i^{dir,t}$
        \State $\hat{p}_{ij}^{t+1} \gets (1-\rho)\hat{p}_{ij}^{t} + \rho \cdot \mathrm{ack}_{ij}^{t}$
    \EndFor
    \State Remove stale neighbors from $N_i^{dir,t}$ if timeout expires
    \If{node $i$ observes addition or removal of edge $e=(u,v)$}
        \State Construct signed claim $c \gets (u,v,\sigma,\nu,\tau,\pi_u)$
        \State Broadcast $c$ to current direct neighbors
    \EndIf
    \For{each received claim $c$ about edge $e$}
        \State \Call{ProcessTopologyClaim}{$i,c,e,t$}
    \EndFor
    \For{each source $j$ with new evidence}
        \State \Call{UpdateSourceTrust}{$i,j,t$}
    \EndFor
    \For{each received claim $c$}
        \State \Call{ControlledForward}{$i,c,t$}
    \EndFor
    \State Perform $E$ local SGD steps to obtain $\theta_i^{t+\frac{1}{2}}$
    \State Exchange candidate models or updates with currently reachable, topology-screened neighbors
     \State Evaluate received peer models on local validation data
    \State Compute validation accuracy $a_{ij}^t$ and mean epistemic uncertainty $\bar{u}_{ij}^t$
    \State Compute \texttt{MURMURA}-style compatibility scores $s_{ij}^{model,t}$
    \For{each currently reachable and topology-screened neighbor $j$}
        \State $q_{ij}^t \gets \lambda_1 s_{ij}^{model,t} + \lambda_2 T_{ij}^{topo,t} + \lambda_3 \hat{p}_{ij}^{t} - \lambda_4 c_{ij}^{comm,t}$
    \EndFor
    \State $C_i^t \gets \operatorname{TopB}_j \; q_{ij}^t$
    \State Aggregate over $C_i^t$ and update $\theta_i^{t+1}$
\EndFor
\end{algorithmic}
\end{algorithm}

\subsection{Direct Reachability and Signed Topology Claims}

The protocol starts from direct observations rather than an assumed global topology. As shown in Algorithm~\ref{alg:main}, each node discovers direct neighbors through authenticated hello beacons and transport-level acknowledgments. For a direct link $(i,j)$, node $i$ updates a local reliability score in the spirit of unreliable link estimation in Soft-DSGD, but without assuming a global reliability matrix \cite{ye2022decentralized}:
\begin{equation}
\hat{p}_{ij}^{t+1} = (1-\rho)\hat{p}_{ij}^{t} + \rho \cdot \text{ack}_{ij}^{t}.
\end{equation}

Here, $\rho \in (0,1)$ is a smoothing parameter and $\text{ack}_{ij}^{t} \in \{0,1\}$ records recent communication success. This quantity captures link quality, not trust \cite{ye2022decentralized}. It provides the local connectivity view on which the later filtering stages build.

When a node observes a local topology event, it emits the signed claim $c$ defined in Section~\ref{sec:local_state}. For direct-neighbor claims, the preferred form is bilateral attestation, meaning both endpoints sign the same event when possible. One-sided claims are allowed but start with lower confidence. Each node may also periodically send a compact signed digest of its current neighbor list (similar to sketching \cite{rangwala2025sketchguard}).

\begin{algorithm}[t]
\caption{Topology claim processing and edge-confidence update}
\label{alg:claim}
\begin{algorithmic}[1]
\State Extract source $j$ and hop count $h_{ij}^t(e)$ from $c$
\If{signature invalid or version stale}
    \State Reject $c$ and return
\EndIf
\State Add source $j$ to witness buffer $\mathcal{W}_i(e)$
\State $\psi_{ij}^t(e) \gets T_{ij}^{topo,t}\exp(-\xi h_{ij}^t(e))$
\If{$c$ supports existence of edge $e$}
    \State Add $\psi_{ij}^t(e)$ to support pool
\Else
    \State Add $\psi_{ij}^t(e)$ to opposition pool
\EndIf
\State $S_i^{+,t}(e) \gets \sum_{j \in \mathcal{W}_i^{+,t}(e)} \psi_{ij}^t(e)$
\State $S_i^{-,t}(e) \gets \sum_{j \in \mathcal{W}_i^{-,t}(e)} \psi_{ij}^t(e)$
\State $\chi_i^t(e) \gets \dfrac{S_i^{+,t}(e)}{S_i^{+,t}(e)+S_i^{-,t}(e)+\varepsilon}$
\If{$c$ is an edge addition}
    \If{$\chi_i^t(e)\ge \tau_{add}$ and witness diversity conditions hold}
        \State Accept edge addition
    \EndIf
\Else
    \If{$1-\chi_i^t(e)\ge \tau_{add}$ and witness diversity conditions hold}
        \State Accept edge removal
    \EndIf
\EndIf
\end{algorithmic}
\end{algorithm}

\subsection{Topology Claim Processing and Edge Confidence}

Received topology claims are not committed immediately. Instead, as shown in Algorithm~\ref{alg:claim}, node $i$ stores each claim in a witness buffer and treats it as weighted evidence about the reported edge. This is the primary defense against malicious topology gossip.

When node $i$ receives a topology claim from source $j$, it computes the witness weight
\begin{equation}
\begin{aligned}
\psi_{ij}^t(e)
&= T_{ij}^{\text{topo},t}
   \exp\!\left(-\xi h_{ij}^t(e)\right) \\
&\quad \cdot \mathbf{1}\{\text{valid signature}\}
       \cdot \mathbf{1}\{\text{fresh version}\}.
\end{aligned}
\end{equation}
Here, $T_{ij}^{\text{topo},t}$ is trust in source $j$, $h_{ij}^t(e)$ is the hop count traveled by the claim before reaching $i$, and $\xi$ penalizes distant claims. Thus, a witness counts more if the source is trusted, the claim traveled fewer hops, the signature is valid, and the claim is fresh. Signed stale claims are ignored.

For each edge $e$, node $i$ maintains separate support and opposition totals:
\begin{align}
S_i^{+,t}(e) &= \sum_{j \in \mathcal{W}_i^{+,t}(e)} \psi_{ij}^t(e), \\
S_i^{-,t}(e) &= \sum_{j \in \mathcal{W}_i^{-,t}(e)} \psi_{ij}^t(e).
\end{align}
The opposition total is therefore the weighted sum of all witnesses that argue against the edge claim. From these totals, node $i$ defines edge confidence as
\begin{equation}
\chi_i^t(e) =
\frac{S_i^{+,t}(e)}{S_i^{+,t}(e)+S_i^{-,t}(e)+\varepsilon}.
\end{equation}

An edge addition is accepted only if three conditions hold:
\begin{itemize}[leftmargin=2em]
\item $\chi_i^t(e) \ge \tau_{add}$, where $\tau_{add}\in(0,1]$ is the \emph{edge-acceptance threshold},
\item support comes from at least $m$ distinct first-hop witnesses,
\item in high-assurance mode, support arrives through at least $k+1$ internally node-disjoint paths when the graph is connected enough to make that meaningful \cite{nesterenko2009discovering}.
\end{itemize}
An edge removal uses the symmetric rule with $1-\chi_i^t(e)$. Thus, a single liar may inject a claim, but that claim is not committed unless enough independent, trusted evidence accumulates.

\begin{algorithm}[t]
\caption{Controlled forwarding}
\label{alg:forward}
\begin{algorithmic}[1]
\State Extract source $j$ from $c$
\If{signature invalid or version stale}
    \State Drop $c$
\ElsIf{$T_{ij}^{topo,t} < \tau_{fwd}$}
    \State Drop $c$
\ElsIf{forwarding budget for source $j$ exhausted}
    \State Drop $c$
\ElsIf{hop count of $c \ge H$}
    \State Drop $c$
\Else
    \State Forward $c$ to eligible direct neighbors
    \State Decrement forwarding budget for source $j$
\EndIf
\end{algorithmic}
\end{algorithm}

\subsection{Controlled Forwarding}

Even a locally observed claim should not be forwarded without restriction. As shown in Algorithm~\ref{alg:forward}, a node forwards a topology claim only if the signature is valid, the version is newer than the local record, the trust score satisfies $T_{ij}^{\text{topo},t} \ge \tau_{\text{fwd}}$ (the \emph{forwarding trust threshold}), the forwarding budget for that source has not been exceeded \cite{gupta2025fully}, and the claim has not crossed a hop limit $H$.

The forwarding budget is important. Gupta et al. \cite{gupta2025fully} show that rate limiting and verification matter even before one considers learning quality. In the present design, repeated flooding of unconfirmed or stale claims contributes negative evidence in the trust update, so forwarding control not only limits immediate spread but also reduces the long-run influence of abusive sources.

\begin{algorithm}[t]
\caption{Source-trust update}
\label{alg:trust}
\begin{algorithmic}[1]
\State Observe evidence events $d_{ij}^t,c_{ij}^t,f_{ij}^t,x_{ij}^t,q_{ij}^t,z_{ij}^t$
\State $\alpha_{ij}^{t+1} \gets \lambda \alpha_{ij}^{t} + w_d d_{ij}^t + w_c c_{ij}^t + w_f f_{ij}^t$
\State $\beta_{ij}^{t+1} \gets \lambda \beta_{ij}^{t} + w_x x_{ij}^t + w_q q_{ij}^t + w_s z_{ij}^t$
\State $R_{ij}^{t+1} \gets \dfrac{\alpha_{ij}^{t+1}}{\alpha_{ij}^{t+1}+\beta_{ij}^{t+1}}$
\State $U_{ij}^{t+1} \gets \sqrt{\dfrac{\alpha_{ij}^{t+1}\beta_{ij}^{t+1}}
{(\alpha_{ij}^{t+1}+\beta_{ij}^{t+1})^2(\alpha_{ij}^{t+1}+\beta_{ij}^{t+1}+1)}}$
\State $T_{ij}^{topo,t+1} \gets R_{ij}^{t+1}\exp\!\left(-\eta \max\{0,U_{ij}^{t+1}-\tau_U\}\right)$
\end{algorithmic}
\end{algorithm}

\subsection{Source Trust Update}

The claim-processing stage in Algorithm~\ref{alg:claim} relies on the current trust score $T_{ij}^{\text{topo},t}$, so this quantity must be updated explicitly over time. As shown in Algorithm~\ref{alg:trust}, node $i$ maintains trust in each source $j$ with a Beta-style evidence tracker.

Positive evidence for $j$ comes from three events: directly observed claims from $j$ about incident edges later match local handshakes, claims from $j$ are corroborated by independent trusted witnesses, and claims from $j$ remain fresh and internally consistent over time. Negative evidence comes from three events: a claim from $j$ is contradicted by later direct observation, a claim from $j$ is contradicted by a trusted quorum, or $j$ floods stale or conflicting topology updates.

A concrete update is
\begin{align}
\alpha_{ij}^{t+1} &= \lambda \alpha_{ij}^t + w_d d_{ij}^t + w_c c_{ij}^t + w_f f_{ij}^t, \\
\beta_{ij}^{t+1}  &= \lambda \beta_{ij}^t + w_x x_{ij}^t + w_q q_{ij}^t + w_s z_{ij}^t,
\end{align}
where $\lambda \in (0, 1]$ is a forgetting factor, $d$ is direct confirmation, $c$ is corroboration, $f$ is freshness consistency, $x$ is contradiction, $q$ is spam or rate violation, and $z$ is severe conflict with the eventually committed local view.

The posterior mean reputation is
\begin{equation}
R_{ij}^{t+1} =
\frac{\alpha_{ij}^{t+1}}{\alpha_{ij}^{t+1} + \beta_{ij}^{t+1}}.
\end{equation}
The epistemic uncertainty of this reputation is taken from the Beta variance:
\begin{equation}
U_{ij}^{t+1} =
\sqrt{
\frac{\alpha_{ij}^{t+1}\beta_{ij}^{t+1}}
{(\alpha_{ij}^{t+1}+\beta_{ij}^{t+1})^2(\alpha_{ij}^{t+1}+\beta_{ij}^{t+1}+1)}
}.
\end{equation}
The final topology trust score mirrors \texttt{MURMURA}’s uncertainty penalization principle \cite{rangwala2025evidential}: the reputation mean is penalized when epistemic uncertainty is high,
\begin{equation}
T_{ij}^{\text{topo},t+1}
=
R_{ij}^{t+1} \cdot
\exp\!\left(-\eta \max\{0, U_{ij}^{t+1}-\tau_U\}\right).
\end{equation}

A node is trusted not only when its claims are often correct, but when that correctness is established with enough evidence. A malicious node that alternates between truthful and false claims may retain a moderate mean reputation, but its uncertainty remains high, keeping its topology influence small.

\subsection{Distributed Local Training and Trust-Aware Collaborator Selection}
\label{sec:collab_score}

The previous modules determine which topology information is credible and which neighbors remain eligible for collaboration. The learning procedure itself is decentralized and runs in parallel across nodes. As shown in Algorithm~\ref{alg:main}, each node performs local training at every epoch, exchanges models or updates only with currently reachable neighbors that survive topology filtering, and then selects a sparse collaborator set for personalized aggregation.

At epoch $t$, node $i$ first performs $E$ local optimization steps on its own data to obtain an intermediate model
\begin{equation}
\theta_i^{t+\frac{1}{2}} = \mathrm{LocalUpdate}(\theta_i^t,\mathcal{D}_i,E),
\end{equation}
where $\mathcal{D}_i$ is the local dataset. After topology filtering, node $i$ exchanges candidate models or updates only with currently reachable and topology-screened neighbors. For those neighbors, it computes the \texttt{MURMURA} model compatibility score $s_{ij}^{\text{model},t}$ from local validation accuracy and epistemic uncertainty \cite{rangwala2025evidential}. Let $a_{ij}^t$ denote the validation accuracy of peer $j$'s model on node $i$'s local validation data, and let $\bar{u}_{ij}^t$ denote the corresponding mean epistemic uncertainty. Following \texttt{MURMURA}, the epistemic uncertainty for a $K$-class evidential model is
\begin{equation}
u = \frac{K}{S},
\end{equation}
where $S=\sum_{k=1}^K \alpha_k$ is the Dirichlet strength. Averaging this quantity over node $i$'s validation samples for peer $j$ gives $\bar{u}_{ij}^t$.

The base compatibility score is then defined as
\begin{equation}
s_{ij,\text{base}}^t
=
\left(1-\bar{u}_{ij}^t\right)
\left(w_a a_{ij}^t + (1-w_a)\right),
\end{equation}
where $w_a \in [0,1]$ controls the contribution of validation accuracy.

The final \texttt{MURMURA}-style compatibility score applies an uncertainty penalty when the mean uncertainty exceeds a threshold $\tau_u$:
\begin{equation}
s_{ij}^{\text{model},t}
=
\begin{cases}
s_{ij,\text{base}}^t \exp\!\left(-(\bar{u}_{ij}^t-\tau_u)\right), & \bar{u}_{ij}^t > \tau_u, \\[4pt]
s_{ij,\text{base}}^t, & \bar{u}_{ij}^t \le \tau_u.
\end{cases}
\end{equation}

The final collaboration score is
\begin{equation}
q_{ij}^t =
\lambda_1 s_{ij}^{\text{model},t}
+ \lambda_2 T_{ij}^{\text{topo},t}
+ \lambda_3 \hat{p}_{ij}^t
- \lambda_4 c_{ij}^{\text{comm},t}.
\end{equation}
The collaborator set is then
\begin{equation}
C_i^t = \operatorname{TopB}_j \; q_{ij}^t,
\end{equation}
subject to an edge budget $B$, as in budgeted decentralized personalized graph learning \cite{kharrat2024decentralized}. Node $i$ aggregates only the intermediate models or updates received from neighbors in $C_i^t$, together with its own model, using \texttt{MURMURA}-style trust-aware weighting:
\begin{equation}
\theta_i^{t+1}
=
w_{ii}^t \theta_i^{t+\frac{1}{2}}
+
\sum_{j \in C_i^t} w_{ij}^t \theta_j^{t+\frac{1}{2}},
\qquad
\sum_{j \in C_i^t \cup \{i\}} w_{ij}^t = 1,
\end{equation}
where the aggregation weights depend on the selected collaboration scores and the \texttt{MURMURA} trust-normalization rule. The collaborator set $C_i^t$ in this final step is the trust-screened set: a candidate $j$ enters $C_i^t$ only if it passes the model-compatibility gate $s_{ij}^{\mathrm{model},t} \ge \gamma\, s_{ii}^{\mathrm{model},t}$, the topology-trust gate $T_{ij}^{\mathrm{topo},t}\ge\tau_{\mathrm{trust}}$, and a lightweight norm-sanity check $\lVert\theta_j^{t+\frac12}\rVert \le \kappa_{\mathrm{norm}}\lVert\theta_i^{t+\frac12}\rVert$ that rejects large-norm model-poisoning outliers; combined with the bounded-domain Assumption~9, this gate directly enforces the bounded-Byzantine-norm condition of Assumption~8 with $\Theta_B = \kappa_{\mathrm{norm}} R$. Crucially, this screening enters the \emph{aggregation weights}, not merely the choice of which peers to contact: a peer that survives contact but fails screening receives weight zero, which is what drives the surviving Byzantine weight $\delta_{\max}$ to zero in practice.

This final stage keeps the topology and model layers distinct: $T_{ij}^{\text{topo},t}$ answers whether topology information from $j$ should be believed, $\hat{p}_{ij}^t$ whether the link is currently usable, and $s_{ij}^{\text{model},t}$ whether $j$ is useful for personalization. All nodes train in parallel, but only credible, reachable, and useful neighbors influence the local model.

\section{Distributed System Implementation}
\label{sec:implementation}

The theoretical protocol described in Section~\ref{sec:protocol} requires a concrete execution substrate: independent client processes that exchange messages, maintain local state, and synchronise rounds without any central coordinator.
This section describes the system architecture constructed to run DMTT experiments on real hardware, and to serve as a reusable open-source extension of the \texttt{MURMURA} framework.

\subsection{Architecture Design Principles} Three principles guided our design. \paragraph{\textbf{No coordinator in the learning path.}} A central coordinator introduces a single point of failure and contradicts the DFL threat model, since a compromised coordinator could suppress or fabricate round signals, undermining the decentralised guarantee. To eliminate this dependency, round $k$ at node $i$ begins at the absolute wall-clock time \begin{equation} t_k = t_{\mathrm{start}} + k \cdot \Delta, \label{eq:wall_clock} \end{equation} where $t_{\mathrm{start}}$ is a shared epoch announced once at experiment start and $\Delta$ is the round budget. Each node independently sleeps until $t_k$ and then proceeds without waiting for any signal. A passive monitor, implemented as a PULL-only process, collects per-round accuracy and loss metrics without sending any messages. Because it has no send socket, it cannot influence the learning protocol, and its failure leaves training unaffected. 

\paragraph{\textbf{ZeroMQ P2P transport.}} Model weights are exchanged peer-to-peer using ZeroMQ PUSH/PULL socket pairs. Each node $i$ binds one PULL socket (its receiving endpoint) and maintains one PUSH socket per current collaborator. This design avoids the fixed process-group requirement of \texttt{torch.distributed} and permits the dynamic collaborator set $C_i^t$ to change every round without re-initialisation. For single-machine experiments, all sockets use IPC transport; for multi-machine deployments, the same code switches to TCP by changing a single configuration field.

\paragraph{\textbf{Deterministic shared mobility model.}} The time-varying graph $\mathcal{G}^t$ must be consistent across all node processes without inter-process communication. This is achieved by defining $\mathcal{G}^t$ through a random-walk mobility model parameterised by a shared seed: given the same seed, every process computes identical node positions and therefore identical edge sets at every round.

\subsection{ZeroMQ Socket Layout}
\label{sec:zmq_layout}

\begin{figure*}[t]
\centering
\includegraphics[width=\textwidth]{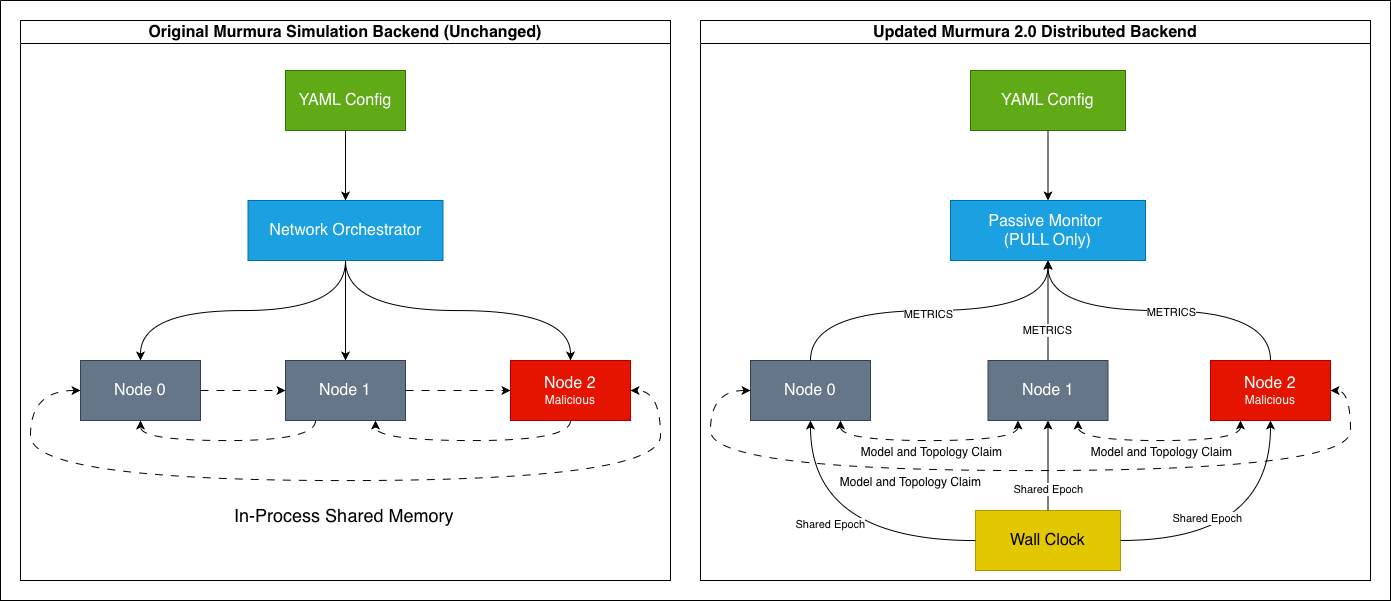}
\caption{Updated \texttt{MURMURA} architecture. The unchanged simulation backend (single OS process, shared memory) and the new distributed backend, in which each client and the passive monitor run as independent OS processes. Model weights travel via \texttt{MODEL\_STATE} PUSH/PULL pairs; DMTT nodes additionally exchange signed neighbourhood observations as \texttt{TOPO\_CLAIM} messages. Round boundaries are set by a shared wall-clock epoch; no coordinator is involved. Byzantine nodes send both poisoned model weights and falsified topology claims.}
\label{fig:architecture}
\end{figure*}

Each node process uses three logical channels, as shown in Fig.~\ref{fig:architecture}.

\begin{itemize}[leftmargin=2em]
  \item \textbf{PULL (bind):} Receives \texttt{MODEL\_STATE} and \texttt{TOPO\_CLAIM} messages from neighbours. A single socket handles both message types, distinguished by a 1-byte type field in the header.
  \item \textbf{PUSH $\times |C_i^t|$ (connect):} Sends \texttt{MODEL\_STATE} and \texttt{TOPO\_CLAIM} to each current collaborator. Sockets are created on first use so that the collaborator set can change without re-binding.
  \item \textbf{PUSH (monitor, connect):} Sends \texttt{METRICS} to the passive monitor after each round evaluation. The monitor never replies.
\end{itemize}

All messages use a two-frame ZMQ multipart format:
\begin{equation*}
  \underbrace{[\,\text{msg\_type}\,(1\text{ B})\;\|\;\text{sender\_id}\,(4\text{ B})\,]}_{\text{frame 0 (header)}}
  \;\Big|\;
  \underbrace{[\,\text{payload}\,]}_{\text{frame 1}}
\end{equation*}
Three message types are defined: \texttt{MODEL\_STATE} (serialised PyTorch state dict), \texttt{METRICS} (pickled evaluation dict containing a \texttt{round\_idx} key for the monitor's ordering buffer), and \texttt{TOPO\_CLAIM} (pickled neighbourhood assertion used exclusively by DMTT nodes).

\subsection{Mobility Model}
\label{sec:mobility_impl}

Client mobility is modelled as a bounded random walk on a 2-D torus of side length $A$.
At each round, the position of client $i$ is updated as
\begin{equation}
  \mathbf{r}_i^{t+1} = \bigl(\mathbf{r}_i^t + \boldsymbol{\delta}_i^t\bigr) \bmod A,
  \label{eq:random_walk}
\end{equation}
where $\boldsymbol{\delta}_i^t \sim \mathrm{Uniform}([-v_{\max}, v_{\max}]^2)$ is an independent per-round displacement drawn from a shared random generator initialised with the shared seed.
An edge $(i,j) \in \mathcal{E}^t$ exists whenever the torus distance between $\mathbf{r}_i^t$ and $\mathbf{r}_j^t$ falls below the communication range $r_c$:
\begin{equation}
  (i,j)\in\mathcal{E}^t \iff d_{\mathrm{torus}}(\mathbf{r}_i^t,\mathbf{r}_j^t) < r_c.
  \label{eq:edge_condition}
\end{equation}
An optional connectivity guarantee connects any isolated node to its nearest peer, ensuring that Assumption~4 (uniform connectivity) is not violated by unlucky mobility realisations.

Because both the initial positions and the displacement sequence are derived deterministically from the shared seed, every node process reconstructs $\mathcal{G}^t$ locally without any out-of-band communication, at the cost of maintaining one random-generator state per round.

\subsection{DMTTNodeProcess Round Protocol}

The \texttt{DMTTNodeProcess} class instantiates Algorithm~\ref{alg:main} as a concrete ZeroMQ process. Each round it: sleeps until the wall-clock deadline $t_k$; trains locally; constructs and pushes \texttt{MODEL\_STATE} + \texttt{TOPO\_CLAIM} to the previous collaborator set; collects responses until the round boundary; updates link-reliability EMAs; evaluates received models to compute $s_{ij}^{\mathrm{model},t}$; applies the oracle-free self-inclusion check---if $j$ communicated with $i$ this round then $j$'s claim must include $i$---to accumulate $d_{ij}^t$ (confirmations) and $x_{ij}^t$ (contradictions) for the Beta trust update; applies the trust-screened aggregation rule of Section~\ref{sec:collab_score}; updates $C_i^t$; and pushes evaluation metrics to the passive monitor. The dominant per-round overhead relative to vanilla NodeProcess is $|C_i^{t-1}|$ forward passes for foreign-model evaluation, scaling as $\mathcal{O}(|C_i^{t-1}| \cdot |\mathcal{D}_i^{\mathrm{test}}| / b)$.

\subsection{Topology-Liar Attack Implementation}

The \texttt{TopologyLiarAttack} concretises the adversary model of Section~\ref{sec:adversary}.
A Byzantine node $b$ mounts an \emph{omission attack}: it sends a falsified TOPO\_CLAIM containing only other Byzantine coalition members and no honest neighbours,
\begin{equation}
  \widehat{\mathcal{N}}_b^t = \mathcal{B} \setminus \{b\},
  \label{eq:liar_claim}
\end{equation}
where $\mathcal{B}$ is the set of all Byzantine nodes.
The omission is detectable without any shared topology knowledge. Because node $i$ sent a \texttt{MODEL\_STATE} to $b$ this round, the directed link $(i,b)$ is established at the transport layer; consequently $i$ is entitled to expect its own identity in $b$'s claim. Since $i\notin\widehat{\mathcal{N}}_b^t$, the self-inclusion check fires a contradiction $x_{ib}^t \mathrel{+}= 1$ at every round, gradually raising $\beta_{ib}$ and suppressing $T_{ib}^{\mathrm{topo},t}$ below the gate threshold $\tau_{\mathrm{trust}}$.

In addition to topology falsification, the liar optionally applies a Gaussian model-state attack (configurable via \texttt{model\_attack\_type: gaussian}), matching the combined threat model described in Section~\ref{sec:adversary}.

\section{Convergence Analysis and Byzantine Bounds}
\label{sec:theory}

This section presents the theoretical analysis of DMTT. The goal is not to prove exact Byzantine topology discovery. Instead, the analysis focuses on the learning dynamics induced by the trusted feasible graph maintained by each client. The presented analysis is consistent with existing works on decentralized optimization over time-varying graphs \cite{koloskova2020unified, nedic2014distributed, 7405263}, topology--heterogeneity coupling \cite{le2023refined}, and Byzantine-resilient decentralized learning \cite{he2022byzantine, 10208131}. All proofs are deferred to the appendix.

\subsection{Effective collaboration matrix}

Continuing the notation of Sections~\ref{sec:problem}--\ref{sec:protocol}: $W^t=[w_{ij}^t]_{i,j=1}^N$ is the screened mixing matrix induced by the active collaborator sets $\{C_i^t\}$, with $w_{ij}^t=0$ whenever $j\notin C_i^t\cup\{i\}$. DMTT is thus decentralized learning over a time-varying screened graph.

In personalized FL, perfect screening severs edges between clients with highly dissimilar data, partitioning the network into disjoint components of compatible clients. To analyze convergence under perfect screening (Theorems 1 and 2), we focus on an arbitrary isolated component $\mathcal{C}$; letting $N = |\mathcal{C}|$ we define the component-average iterate $\bar\theta^t = \tfrac{1}{N}\sum_{i\in\mathcal C}\theta_i^t$ and the component objective $f(x) = \tfrac{1}{N}\sum_{i\in\mathcal C} F_i(x)$. This correctly reflects the personalized DFL setting: the protocol filters the graph so that clients optimize their compatible component objective $f$ rather than the true global average \cite{koloskova2020unified}.

\subsection{Standing assumptions and theory-local notation}
\label{sec:standing_assumptions}

Beyond the symbols of Sections~\ref{sec:problem}--\ref{sec:protocol}, the analysis uses $\widetilde W^t$ (ideal honest mixing matrix), $p\in(0,1]$ (consensus-rate parameter of the screened graph sequence \cite{koloskova2020unified}), $\tau_{\mathrm{scr}}^2$ (screened neighborhood heterogeneity), $\Delta_{\mathrm{scr}}^2$ (screening-error level), $\delta_{\max}$ (maximum surviving Byzantine weight after screening), $G$ (uniform stochastic-gradient second-moment bound, introduced with Assumption~6), and $\Theta_B$ (uniform Byzantine-message norm bound, introduced with Assumption~8), with $M := G + L\Theta_B$.

We collect below the four assumptions required by the consensus and descent lemmas (deferred to the appendix); Theorems and propositions later in the section introduce additional assumptions only where they are first needed. The link-reliability update, witness-weight rule, Beta trust update, edge-confidence rule, and collaborator score are treated as part of the DMTT mechanism, not as standalone assumptions.

\paragraph{Assumption 1: Smooth local objectives.}
For each client $i$, the function $F_i$ is $L$-smooth:
$$
\|\nabla F_i(x)-\nabla F_i(y)\| \le L\|x-y\|,
\qquad
\forall x,y\in\mathbb R^d.
$$

\paragraph{Assumption 2: Unbiased stochastic gradients with bounded variance.}
Each client has access to a stochastic gradient oracle $g_i(x,\xi_i)$ such that
$$
\mathbb E[g_i(x,\xi_i)] = \nabla F_i(x),
\quad
\mathbb E\|g_i(x,\xi_i)-\nabla F_i(x)\|^2 \le \sigma^2.
$$

\paragraph{Assumption 3: Doubly stochastic screened mixing.}
At each epoch $t$, the DMTT aggregation step produces a doubly stochastic matrix $W^t = [w_{ij}^t]$ satisfying
$$
W^t\mathbf 1 = \mathbf 1,
\quad
\mathbf 1^\top W^t = \mathbf 1^\top,
\quad
w_{ij}^t \ge 0,
$$
with the screened-support property $w_{ij}^t = 0$ if $j\notin C_i^t\cup\{i\}$. In practice, this can be enforced by symmetrizing the trust-weighted scores, or by adopting a Metropolis--Hastings reweighting over the screened graph. A fully asymmetric collaboration matrix is also supported, but would require a push-sum or gradient-tracking treatment in place of the analysis below \cite{nedic2014distributed, 7405263}.

\paragraph{Assumption 4: Consensus contraction of the screened honest graph.}
Let $J = \tfrac{1}{N} \mathbf{1} \mathbf{1}^\top$. There exists $p \in (0,1]$ such that for every $t \ge 0$ and every $X \in \mathbb{R}^{N \times d}$ with $\mathbf{1}^\top X = 0$,
$$
\| (W^t - J) X \|_F^2 \le (1-p) \|X\|_F^2.
$$
The constant $p$ is the consensus-rate parameter of the screened graph sequence. This per-step contraction is stronger than the $B_{\mathrm{conn}}$-window connectivity formulation of \cite{koloskova2020unified}; the extension to the window case follows the same template with constants scaling as $1/p \to B_{\mathrm{conn}}/p$.

The convergence proof relies on standard structural invariants, a consensus contraction bound, and a descent lemma over the screened graph. These intermediate lemmas are detailed in the appendix.

\subsection{Main theorem under perfect screening}

\emph{Perfect screening} means that the topology layer correctly excludes all Byzantine clients ($\delta_{\max} = 0$) and produces the ideal honest mixing matrix ($\Delta_{\mathrm{scr}} = 0$). Under this condition, DMTT reduces to time-varying decentralized SGD on the screened honest graph, with one local SGD step per round ($E = 1$). One additional assumption is needed beyond the standing ones: heterogeneity must be uniformly bounded across screened collaborators.

\paragraph{Assumption 5: Bounded screened heterogeneity.}
There exists $\tau_{\mathrm{scr}}^2 \ge 0$ such that, for all $x \in \mathbb{R}^d$,
$$
\frac{1}{N} \sum_{i \in \mathcal{C}} \big\| \nabla F_i(x) - \nabla f(x) \big\|^2 \le \tau_{\mathrm{scr}}^2.
$$
The subscript ``scr'' indicates that this constant is measured over the population that participates in screened aggregation under DMTT. Corollary~\ref{cor:screening} discusses how \texttt{MURMURA}-style screening can reduce this constant relative to the heterogeneity over the raw feasible graph \cite{le2023refined}.

\begin{theorem}[Convergence under perfect screening, $E = 1$]
\label{thm:main}
Suppose Assumptions~1--5 hold, screening is perfect ($\delta_{\max} = 0$ and $\Delta_{\mathrm{scr}}^2 = 0$), and the iterates are initialized at consensus ($\theta_i^0 = \bar\theta^0$ for all $i$). If the stepsize satisfies
$$
\eta \le \frac{p}{8L},
$$
then for any horizon $T \ge 1$,
\begin{equation}
\label{eq:main_bound}
\begin{aligned}
\frac{1}{T} \sum_{t=0}^{T-1} \mathbb{E} \big\| \nabla f(\bar\theta^t) \big\|^2
\le\;& \frac{4 \big(f(\bar\theta^0) - f^\star\big)}{\eta T}
+ \frac{2 L \eta\, \sigma^2}{N} \\
& + \frac{64\, L^2\, \eta^2 \big(\tau_{\mathrm{scr}}^2 + \sigma^2\big)}{p^2},
\end{aligned}
\end{equation}
where $f^\star = \inf_x f(x)$. The proof is given in the appendix.
\end{theorem}

\begin{remark}[Interpretation of the rate]
The three terms in~\eqref{eq:main_bound} are the optimization error from initialization ($\mathcal O(1/T)$), the stochastic-noise floor ($\propto \sigma^2/N$), and the heterogeneity-and-topology term ($\propto (\tau_{\mathrm{scr}}^2 + \sigma^2)/p^2$). Choosing $\eta = \min\!\big\{\tfrac{p}{8L},\, \sqrt{\tfrac{N (f^0 - f^\star)}{L \sigma^2 T}}\big\}$ yields an asymptotic convergence rate of $\mathcal{O}\!\big(\frac{1}{\sqrt{N T}} + \frac{N(\tau_{\mathrm{scr}}^2 + \sigma^2)}{p^2 T}\big)$. While state-of-the-art analyses can tighten the second term to $\mathcal{O}(1/(pT))$ using complex potential functions \cite{koloskova2020unified}, this standard analysis suffices to explicitly show how the DMTT screening filter accelerates convergence by reducing the topological heterogeneity $\tau_{\mathrm{scr}}^2$.
\end{remark}

\begin{remark}
In the appendix, we extend this result to $E \ge 1$ local steps (Theorem 2) under an additional bounded-gradient assumption ($\mathbb{E}\|g_i\|^2 \le G^2$); local updates introduce an $\mathcal{O}(\eta^2 L^2 G^2 / p^2)$ drift penalty but reduce the stochastic noise floor to $\sigma^2 / (NE)$.
\end{remark}

\subsection{Why screening helps}
\label{sec:why_screening}
The bound~\eqref{eq:main_bound} shows that convergence depends heavily on the data heterogeneity $\tau_{\mathrm{scr}}^2$. DMTT explicitly shrinks this constant by severing edges to model-incompatible peers (large $\bar u_{ij}^t$). While dropping edges may reduce the graph's consensus rate $p$, the reduction in heterogeneity typically dominates.

\begin{corollary}[Screening improves the bound]
\label{cor:screening}
Let $\mathcal{V}_{\mathrm{raw}}$ be the full, unfiltered network with objective $f_{\mathrm{raw}}$, resulting in global heterogeneity $\tau_{\mathrm{raw}}^2$ and consensus rate $p_{\mathrm{raw}}$. Suppose the DMTT screening mechanism isolates a highly compatible component $\mathcal{C}$, shifting the target objective to $f$ as defined above. If the screening yields
$$
\tau_{\mathrm{scr}}^2 \le \tau_{\mathrm{raw}}^2 \quad \text{and} \quad p_{\mathrm{scr}} \ge p_{\mathrm{raw}} - \epsilon
$$
for some $\epsilon \ge 0$, then the heterogeneity-and-topology term in~\eqref{eq:main_bound} satisfies
$$
\frac{64 L^2 \eta^2 \big(\tau_{\mathrm{scr}}^2 + \sigma^2\big)}{p_{\mathrm{scr}}^2} \le \frac{64 L^2 \eta^2 \big(\tau_{\mathrm{raw}}^2 + \sigma^2\big)}{(p_{\mathrm{raw}} - \epsilon)^2}.
$$
By filtering out incompatible peers, DMTT deliberately abandons the global stationary point in favor of a personalized component stationary point, achieving a tighter convergence bound as long as the reduction in data heterogeneity outpaces the loss in subgraph connectivity.
\end{corollary}

\subsection{Perturbation under imperfect screening}

We close the analysis with the bound that quantifies the cost of imperfect screening. The proposition makes precise the two terms that appear when (i) the screened matrix differs from the ideal honest mixing ($\Delta_{\mathrm{scr}} > 0$), and (ii) Byzantine peers retain residual aggregation weight ($\delta_{\max} > 0$). Three additional assumptions are needed: a bounded second moment of the stochastic gradient (which strengthens Assumption~2 and controls the local drift used in the perturbation term), a bound on the screening error and surviving Byzantine influence, and a bound on the norm of Byzantine messages. Without the last, Byzantine influence cannot be bounded in general because adversarial peers can broadcast arbitrarily large updates.

\paragraph{Assumption 6: Bounded gradient second moments.}
There exists $G \ge 0$ such that for all $i \in \mathcal V$, $t \ge 0$, and points $x$ visited by the iterates,
$$
\mathbb{E}\| g_i(x, \xi_i) \|^2 \le G^2.
$$
This is stronger than Assumption~2 (bounded \emph{variance}); note $G^2 \ge \sigma^2 + \sup_x \|\nabla f(x)\|^2$. In practice, uniform Lipschitz regularization or gradient clipping enforces it. Assumption~6 is invoked in Proposition~\ref{prop:perturbation} below (through $\eta G$ in the drift magnitude $\rho$) and in the extension to $E > 1$ local steps (Theorem~2, appendix).

\paragraph{Assumption 7: Bounded screening error and surviving Byzantine influence.}
Let $\widetilde W^t$ denote the ideal honest mixing matrix induced by the true feasible graph, and let $W^t$ denote the actual DMTT screened matrix. Assume
$$
\mathbb E\|W^t-\widetilde W^t\|_F^2 \le \Delta_{\mathrm{scr}}^2.
$$
Also define
$$
\delta_{\max}
=
\sup_t \max_{i\in\mathcal H}\sum_{j\in\mathcal B} w_{ij}^t,
$$
and assume $\delta_{\max}$ is uniformly bounded and sufficiently small relative to the connectivity of the screened honest graph.

\paragraph{Assumption 8: Bounded Byzantine message norms.}
There exists $\Theta_B \ge 0$ such that, for every Byzantine client $j \in \mathcal B$ and every epoch $t$,
$$
\big\| \theta_j^{B, t+\tfrac12} \big\| \le \Theta_B
$$
almost surely, where $\theta_j^{B, t+\tfrac12}$ denotes the (possibly adversarial) value broadcast by Byzantine peer $j$ in round $t$. This assumption is invoked only in Proposition~\ref{prop:perturbation}. It is \emph{not} a consequence of Assumptions~1--7. The norm-sanity gate of Section~\ref{sec:collab_score} provides a concrete enforcement: any peer $j$ that survives to aggregation satisfies $\|\theta_j\| \le \kappa_{\mathrm{norm}}\|\theta_i\|$. Under Assumption~9 ($\|\theta_i\|\le R$ for all $i,t$), every surviving Byzantine peer satisfies $\|\theta_j\|\le \kappa_{\mathrm{norm}} R$, giving $\Theta_B = \kappa_{\mathrm{norm}} R$ as a hard bound for peers not dropped by the gate.

\paragraph{Assumption 9: Bounded domain.}
The local iterates are confined to a bounded domain such that for all $i \in \mathcal{C}$ and $t \ge 0$, $\|\theta_i^t\| \le R$. In practice, this is enforced via projected SGD or $\ell_2$ weight decay.

\begin{proposition}[Imperfect screening]
\label{prop:perturbation}
Suppose Assumptions~1--9 hold over the screened component $\mathcal{C}$ of size $N$. Let $\widetilde W^t$ be the ideal doubly stochastic matrix over $\mathcal{C}$, and let $W^t$ be the actual DMTT submatrix. Assume $\mathbb E\|W^t-\widetilde W^t\|_F^2 \le \Delta_{\mathrm{scr}}^2$ and let $\delta_{\max}$ bound the total weight any honest node assigns to Byzantine peers. Define the worst-case per-step drift magnitude:
$$
\rho := \Delta_{\mathrm{scr}} (R + \eta G) + \delta_{\max} \Theta_B.
$$
If $\eta \le \tfrac{p}{24L}$ and $E = 1$, then for any horizon $T \ge 1$:
\begin{equation}
\label{eq:perturbed_bound}
\begin{aligned}
\frac{1}{T} \sum_{t=0}^{T-1} \mathbb{E}\big\| \nabla f(\bar\theta^t) \big\|^2 \le\;&
\frac{4 (f^0 - f^\star)}{\eta T} + \frac{4 L \eta \sigma^2}{N} \\
&+ \frac{256 L^2 \eta^2 (\tau_{\mathrm{scr}}^2 + \sigma^2/4)}{p^2} \\
&+ \frac{48 L^2}{p^2} \rho^2 + \Big(\frac{8}{\eta^2} + \frac{4L}{\eta}\Big)\rho^2.
\end{aligned}
\end{equation}
\end{proposition}

\begin{remark}[Interpretation]
The first three terms of~\eqref{eq:perturbed_bound} match the Byzantine-free rate of Theorem~\ref{thm:main} up to constants; the final term, driven by $\rho^2$, isolates the topology-screening error $\Delta_{\mathrm{scr}}$ and the surviving Byzantine influence $\delta_{\max}$. Because the Byzantine message induces an additive shift in \emph{parameter} space rather than \emph{gradient} space, the penalty scales as $\mathcal{O}(\rho^2/\eta^2)$: unlike standard SGD, shrinking $\eta$ does not drive the error to zero — an arbitrarily small stepsize would let the unscaled Byzantine parameter drift dominate the gradient signal, so a fixed stepsize floor is essential in robust decentralized learning.
\end{remark}

\section{Experimental Evaluation}
\label{sec:experiments}

\subsection{Conditions and Baselines}

The controlled comparison isolates the contributions of dynamic topology and of the DMTT trust protocol through three primary conditions, all subjected to the same combined topology-liar + Gaussian attack:

\begin{enumerate}[leftmargin=2em, label=\textbf{C\arabic*.}]
  \item \textbf{FedAvg (static).} Fixed fully-connected topology; plain FedAvg over all neighbours; no trust. This is the threat faced by unmodified decentralized averaging when the topology is wrongly assumed static.
  \item \textbf{FedAvg (dynamic).} Dynamic topology $\mathcal{G}^t$ from the mobility model; each node averages over all direct neighbours $\mathcal{N}_i^{\mathrm{dir},t}$ with no trust filtering. Comparing C1 and C2 isolates the effect of mobility alone.
  \item \textbf{DMTT.} Same mobility model as C2 with the full DMTT protocol active: link-reliability EMA, Beta source-trust with \texttt{TOPO\_CLAIM} verification, and the trust-screened, $q_{ij}$-weighted aggregation of Section~\ref{sec:collab_score}. Comparing C2 and C3 isolates the benefit of trust.
\end{enumerate}

To position DMTT against the state of the art in Byzantine-resilient decentralized aggregation, we additionally evaluate three robust aggregators over the same dynamic graph and attack: Krum (multi-Krum distance selection)~\cite{blanchard2017machine}, BALANCE (adaptive distance filtering)~\cite{fang2024byzantine}, and UBAR (two-stage distance-plus-loss filtering)~\cite{guo2021byzantine}.
We first validated the distributed ZeroMQ backend against the simulation backend at $10$ nodes across three Byzantine fractions and two seeds each, running DMTT and FedAvg (dynamic) on a real Melbourne Research Cloud (MRC) testbed with wall-clock synchronisation and IPC transport. The two backends produce qualitatively identical convergence trajectories across all tested conditions: DMTT converges to high accuracy in both cases, and FedAvg collapses to chance under attack in both cases. The mean final honest-node accuracy (last 5 of 20 rounds) agrees within $0.07$ across all condition and Byzantine fraction pairs (convergence curves in the appendix). Having established that the simulation faithfully reproduces real distributed behaviour, we ran all main experiments at $100$ nodes on the simulation backend, which executes the protocol logic in a single process, making the full sweep across $8$ Byzantine fractions and $3$ seeds per condition deterministic and tractable.

\subsection{Datasets and Models}

\paragraph{UCI HAR.}
We first evaluate on \textbf{UCI HAR}~\cite{anguita2013public}, the smartphone human-activity-recognition benchmark used in the original \texttt{MURMURA} study: $561$ hand-crafted inertial features per window and six activity classes (walking, walking upstairs, walking downstairs, sitting, standing, lying). The $7{,}352$ training instances are partitioned across $N=100$ mobile clients with a Dirichlet distribution of concentration $\alpha = 0.5$, producing realistic non-IID local distributions; each client holds out $20\%$ of its partition as a local test set on which all reported accuracies are computed. Every client trains a two-hidden-layer perceptron ($256$--$128$ units, dropout $0.3$) with a Dirichlet (evidential) output head, optimised with the evidential loss of \cite{rangwala2025evidential, sensoy2018evidential}. Chance accuracy on the six balanced classes is $\tfrac{1}{6}\approx 0.167$.

\paragraph{PAMAP2.}
To confirm generality we additionally evaluate on \textbf{PAMAP2}~\cite{reiss2012introducing}, a body-worn sensor benchmark with richer temporal structure and finer-grained activity classes: nine subjects performed twelve physical activities (lying, sitting, standing, walking, running, cycling, Nordic walking, ascending/descending stairs, vacuuming, ironing, rope jumping) while wearing inertial measurement units on hand, chest, and ankle ($100\,\mathrm{Hz}$, $13$ features per IMU after excluding invalid orientation columns) together with a wrist heart-rate monitor. Raw sensor streams are segmented into $50$-sample ($0.5$-second) windows with $25$-sample stride; each window is flattened to a $2{,}000$-dimensional feature vector ($50 \times 40$). This representation captures within-activity dynamics while remaining computationally comparable to UCI HAR's pre-engineered $561$-feature windows. The roughly $38{,}800$ windows are partitioned across the same $N=100$ clients using a Dirichlet $\alpha=0.5$ split; each client trains an evidential perceptron ($256$--$128$ units, dropout $0.3$) of the same depth as the UCI HAR model. Because PAMAP2 has $12$ balanced classes, chance accuracy is $\tfrac{1}{12}\approx 0.083$, and any method above $0.50$ is meaningfully learning. The model-compatibility score $s_{ij}^{\mathrm{model}}$ is computed identically on both datasets: local validation accuracy combined with epistemic vacuity uncertainty.

\begin{table}[t]
\centering
\small
\caption{Experimental hyperparameters. Primary metric: mean honest-node accuracy on each client's held-out local test partition, averaged over the final 10 of 50 rounds.}
\label{tab:exp_setup}
\begin{tabular}{@{}p{0.50\columnwidth}p{0.42\columnwidth}@{}}
\toprule
\textbf{Parameter} & \textbf{Value} \\
\midrule
Number of clients $N$           & 100 \\
Datasets                        & UCI HAR, PAMAP2 \\
Heterogeneity                   & Dirichlet $\alpha = 0.5$ \\
Byzantine fraction              & $10$--$80\%$ (10--80 nodes) \\
Attack type                     & Topology liar + Gaussian ($\sigma = 10$) \\
Training rounds                 & 50 \\
Local epochs $E$                & 2 \\
Batch size / learning rate      & $32$ / $0.01$ \\
Seeds                           & $\{42,43,44\}$ \\
\midrule
\multicolumn{2}{@{}l}{\emph{Mobility model (dynamic conditions)}} \\
Arena $A$                       & $100 \times 100$ \\
Communication range $r_c$       & 40 \\
Max speed $v_{\max}$            & 8 per round \\
\midrule
\multicolumn{2}{@{}l}{\emph{DMTT protocol (C3)}} \\
Budget $B$                      & 5 \\
EMA rate $\rho$                 & 0.1 \\
Forgetting factor $\lambda$     & 0.9 \\
Trust weights $(w_d, w_x)$      & $(1.0,\;1.0)$ \\
Uncertainty threshold / penalty $(\tau_U,\eta)$ & $(0.3,\;5.0)$ \\
Score weights $(\lambda_1,\lambda_2,\lambda_3,\lambda_4)$ & $(0.4,\;0.3,\;0.2,\;0.1)$ \\
Model-compat. gate $\gamma$     & $0.6$ \\
Topology-trust gate $\tau_{\mathrm{trust}}$ & $0.25$ \\
Norm-sanity ratio               & $5\times$ \\
Self-weight $\omega$            & $0.5$ \\
\bottomrule
\end{tabular}
\end{table}

\subsection{Results}

\begin{figure*}[ht]
\centering
\includegraphics[width=\textwidth]{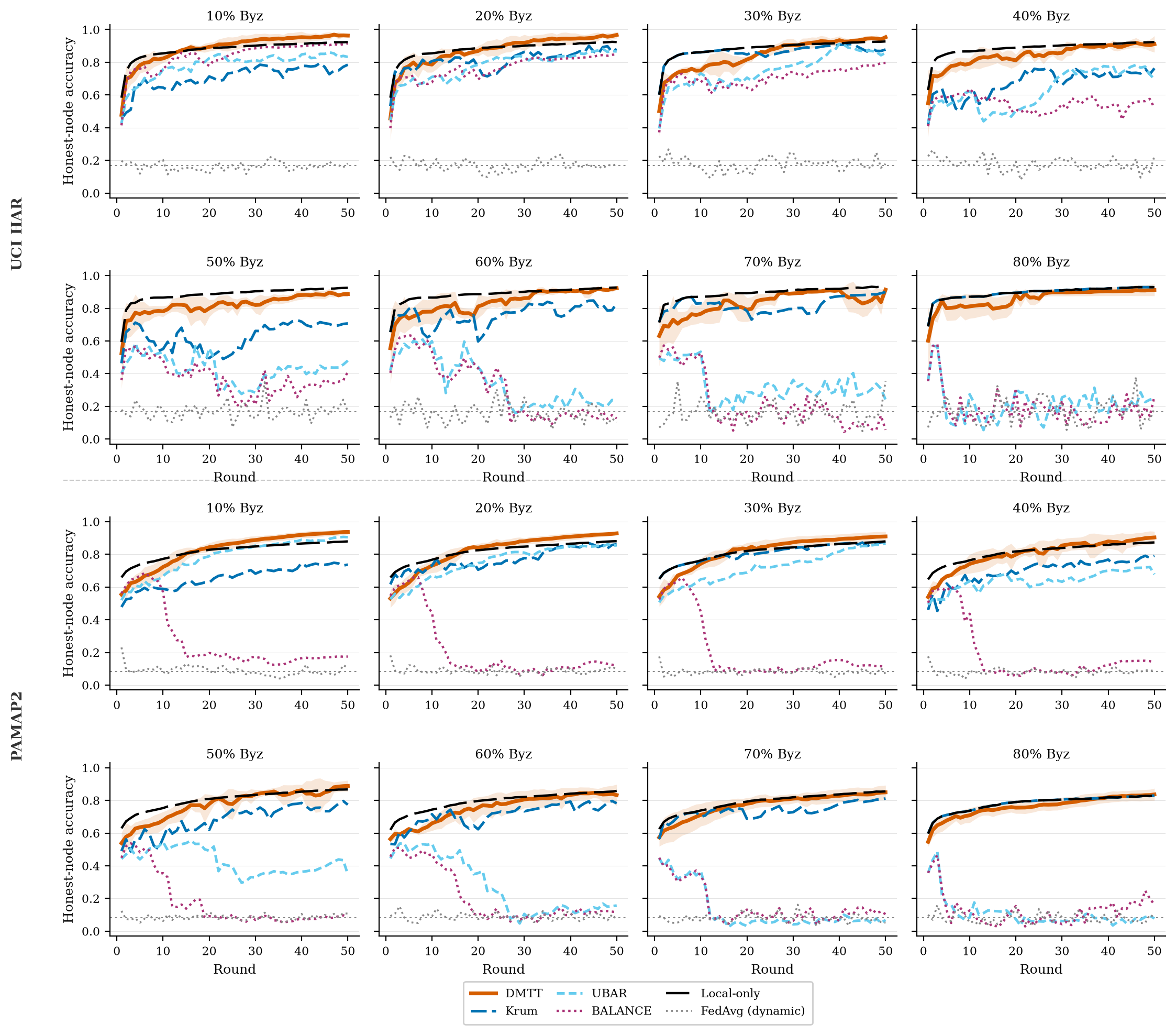}
\caption{Honest-node accuracy versus round under the combined topology-liar + Gaussian ($\sigma{=}10$) attack (mean $\pm$ std over three seeds). \textbf{Top two rows}: UCI HAR ($6$ classes, chance $=0.167$). \textbf{Bottom two rows}: PAMAP2 ($12$ classes, chance $=0.083$). FedAvg collapses to chance at all fractions on both datasets. Krum and BALANCE each fail to consistently clear the local-only bar. UBAR is competitive at low Byzantine fractions but collapses above $50\%$. DMTT sustains high honest-node accuracy across all $10$--$80\%$ fractions on both datasets, remaining at or above local-only throughout.}
\label{fig:dmtt_curves_combined}
\end{figure*}

\subsubsection{UCI HAR}

\paragraph{Trust-aware screening is decisive.}
Figure~\ref{fig:dmtt_curves_combined} (top two rows) shows final honest-node accuracy for all methods across Byzantine fractions from $10\%$ to $80\%$ (mean $\pm$ std over three seeds, averaged over the last $10$ of $50$ rounds). Both non-robust baselines collapse to chance at all fractions: \textbf{FedAvg (static)} and \textbf{FedAvg (dynamic)} never exceed $0.22$, confirming that even a single $\sigma=10$ Gaussian-poisoned model averaged into the aggregate destroys learning, regardless of graph type. \textbf{DMTT} instead sustains honest accuracy above $0.862$ at $10\%$--$80\%$ Byzantine, approximately a $5\times$ improvement over non-robust baselines at all fractions.

\paragraph{Graceful degradation at extreme Byzantine fractions.}
At low-to-moderate fractions ($10$--$30\%$ Byzantine), DMTT significantly outperforms \textbf{Local-only} ($0.920$--$0.923$), demonstrating that DMTT extracts genuine collaborative benefit from honest neighbours while screening adversaries. As the Byzantine fraction rises above $40\%$, the number of reachable honest neighbours per node in the dynamic graph diminishes, so the screened aggregate converges toward the local model and the collaboration surplus over local-only narrows. Accuracy is non-monotone in the Byzantine fraction: DMTT dips to $0.862 \pm 0.061$ at $70\%$ before recovering to $0.908 \pm 0.031$ at $80\%$. This reflects the fallback-to-local-only mechanism: at $80\%$ Byzantine, trust gates exclude nearly all adversarial neighbours and the effective aggregate is almost entirely the local model; since local-only accuracy ($0.927$) is high, performance recovers relative to $70\%$. At $70\%$ the honest-neighbour pool is still large enough that the system attempts more collaboration, but the high adversarial density makes clean peer selection harder, yielding a lower mean with higher variance ($\sigma{=}0.061$ vs.\ $0.031$ at $80\%$). At $80\%$ Byzantine ($80$ of $100$ nodes adversarial), DMTT ($0.908 \pm 0.031$) is within one standard deviation of Local-only ($0.927 \pm 0.035$) and clearly superior to BALANCE ($0.191$) and UBAR ($0.225$), confirming the \emph{near-safety} property: even when honest nodes are overwhelmed by adversaries, trust-screened DMTT degrades to approximately local-only performance rather than collapsing to chance.

\paragraph{Dedicated robust aggregators fail to clear the local-only bar.}
By the local-only bar, all three dedicated robust baselines fail. \textbf{Krum} is competitive at $10$--$30\%$ Byzantine ($0.768$--$0.887$) but degrades at $40$--$60\%$ ($0.734$--$0.819$), below both DMTT ($0.880$--$0.907$) and local-only ($0.919$--$0.922$), because topology-liar poisoning causes it to select compromised models that pass the distance filter. At $70\%$ Byzantine, Krum's performance is variable across seeds (ranging from $0.805$ to $0.942$, sometimes exceeding local-only and sometimes falling below it). At $80\%$ Byzantine, Krum's output matches local-only \emph{exactly on every seed} (mean accuracy difference: $0.000$). With $80$ of $100$ nodes adversarial, Multi-Krum's nearest-neighbour count collapses for any realistic neighbourhood size, so Krum selects only the single closest model. Under $\sigma=10$ Gaussian noise the honest models cluster tightly and the selected peer is invariably an honest one, but aggregating with a single peer under Dirichlet $\alpha=0.5$ non-IID data provides negligible collaborative gain, yielding accuracy statistically indistinguishable from local-only. This is the same graceful degradation toward local-only that DMTT exhibits at extreme fractions, reached here through a different mechanism. \textbf{BALANCE} degrades catastrophically ($0.537$ at $40\%$, $0.147$ at $60\%$, $0.078$ at $70\%$), driven by the brittleness of adaptive distance thresholds under simultaneous topology manipulation and non-IID data. \textbf{UBAR} degrades more gradually but still falls below local-only above $50\%$ Byzantine. DMTT is the only method that consistently clears the local-only bar at all fractions from $10\%$ to $80\%$.

\paragraph{Byzantine influence is screened to zero.} All poisoned peers are excluded by the three aggregation gates: $\delta_{\max}=0$ in every round, even at Byzantine fractions up to $80\%$. Under $30\%$ Byzantine participation, honest-peer trust rises to $T_{ij}^{\mathrm{topo}}\approx0.986$ by rounds $40$--$50$, whereas topology liars fall to $0.012$, producing a $0.974$ separation via the oracle-free self-inclusion check.

\subsubsection{PAMAP2}

We next evaluate on PAMAP2 to test whether DMTT's trust mechanisms generalize to richer sensor data with finer-grained activity classes and raw temporal windows, a setting that stresses distance-based aggregators more severely than UCI HAR's pre-engineered features.

\paragraph{Generalisation to richer sensor data.}
Figure~\ref{fig:dmtt_curves_combined} (bottom two rows) shows the same comparison on PAMAP2 ($12$ classes, chance $= 0.083$). Both FedAvg variants collapse to near-chance at all Byzantine fractions ($0.064$--$0.095$, indistinguishable from chance). \textbf{DMTT} sustains honest accuracy above $0.829$ across all Byzantine concentrations, a $9\times$--$13\times$ lift over the non-robust baselines. Crucially, DMTT clears the local-only bar ($0.873$--$0.823$) at \emph{every} Byzantine fraction from $10\%$ to $80\%$; because PAMAP2's twelve-class task is harder, the collaboration benefit persists even at extreme fractions. \textbf{BALANCE} collapses entirely on PAMAP2 (reaching only $0.124$--$0.170$ across all fractions, near-chance on a 12-class task), confirming that its distance thresholds cannot tolerate simultaneous topology manipulation and raw-window feature heterogeneity. \textbf{UBAR} starts competitively ($0.894$ at $10\%$) but collapses at $50\%$ Byzantine (mean $0.395 \pm 0.350$, highly variable) and reaches chance at $70\%$--$80\%$. The Beta trust separation replicates exactly: honest peers reach $T_{ij}^{\mathrm{topo}} \approx 0.997$ while liars are suppressed to $T_{ij}^{\mathrm{topo}} \approx 0.022$ (identical to UCI HAR), confirming that the trust layer is insensitive to the choice of feature representation.

\paragraph{Convergence.}
Figure~\ref{fig:dmtt_curves_combined} (bottom two rows) shows PAMAP2 learning trajectories across all eight Byzantine fractions. \textbf{BALANCE} fails near-chance throughout at every fraction, confirming that its adaptive distance thresholds cannot tolerate simultaneous topology manipulation and raw-window feature heterogeneity. \textbf{UBAR} converges well at low Byzantine fractions but collapses above $50\%$, exhibiting high seed-to-seed variance at that threshold. \textbf{DMTT} maintains clear separation from all baselines at every fraction, remaining at or above the local-only reference throughout. The same qualitative pattern holds as on UCI HAR, but with a sustained collaboration surplus even at extreme fractions due to PAMAP2's twelve-class task difficulty.

\subsection{Hyperparameter Sensitivity}
\label{sec:ablation}

A practical concern for any trust-based protocol is whether performance depends on careful parameter tuning.
Table~\ref{tab:sensitivity} addresses this by varying one DMTT hyperparameter at a time while holding all others at the default, across four Byzantine fractions on both datasets.

The four parameters tested cover the trust dynamics ($\lambda$, the Beta forgetting factor), the link-reliability estimator ($\rho$, the EMA smoothing rate), the topology-trust gate threshold ($\tau_{\mathrm{trust}}$), and the collaborator budget ($B$).
Across all eight variants, honest-node accuracy stays within $\pm 0.005$ of the default at every Byzantine fraction.
The forgetting factor $\lambda$ can be raised from $0.70$ to $0.95$ with no measurable effect, because topology-liar contradictions accumulate rapidly at every decay rate: even with $\lambda=0.70$, a liar that falsifies two neighbours per round reaches a suppressed trust score within a handful of rounds.
The EMA rate $\rho$ is similarly irrelevant, as direct link quality is perfectly observable under the shared mobility model and does not require smoothing on any particular timescale.
The trust gate threshold $\tau_{\mathrm{trust}}$ can be relaxed to $0.10$ or tightened to $0.50$ without loss: Byzantine nodes are driven to $T_{ij}^{\mathrm{topo}} \lesssim 0.02$ after trust convergence, well below even the loose threshold.
Reducing the budget to $B=3$ produces a small drop at $10\%$ Byzantine ($0.956$ vs.\ $0.960$) because the node has fewer honest collaborators available, but the difference closes entirely at higher Byzantine fractions where the feasible collaborator pool is already constrained.

\begin{table}[t]
\centering
\small
\caption{DMTT hyperparameter sensitivity (mean $\pm$ std over 3 seeds, last 10 of 50 rounds, $\alpha{=}0.5$, combined topology-liar + Gaussian attack). One parameter is varied at a time; all others remain at the default.}
\label{tab:sensitivity}
\resizebox{\columnwidth}{!}{%
{\setlength{\tabcolsep}{3pt}%
\begin{tabular}{@{}lcccc@{}}
\toprule
Configuration & 10\% Byz & 30\% Byz & 50\% Byz & 80\% Byz \\
\midrule
\textbf{Default} ($\lambda{=}0.90,\;\rho{=}0.10,\;\tau{=}0.25,\;B{=}5$) & $0.960 \pm 0.014$ & $0.936 \pm 0.015$ & $0.880 \pm 0.018$ & $0.907 \pm 0.025$ \\
\midrule
$\lambda=0.70$ (faster trust decay) & $0.960 \pm 0.015$ & $0.940 \pm 0.012$ & $0.887 \pm 0.012$ & $0.908 \pm 0.025$ \\
$\lambda=0.95$ (slower trust decay) & $0.960 \pm 0.015$ & $0.940 \pm 0.012$ & $0.886 \pm 0.012$ & $0.908 \pm 0.025$ \\
\midrule
$\rho=0.05$ (slower link EMA) & $0.959 \pm 0.015$ & $0.940 \pm 0.012$ & $0.885 \pm 0.010$ & $0.908 \pm 0.025$ \\
$\rho=0.30$ (faster link EMA) & $0.960 \pm 0.015$ & $0.940 \pm 0.012$ & $0.885 \pm 0.011$ & $0.908 \pm 0.025$ \\
\midrule
$\tau_{\mathrm{trust}}=0.10$ (loose gate) & $0.960 \pm 0.015$ & $0.940 \pm 0.012$ & $0.886 \pm 0.012$ & $0.908 \pm 0.025$ \\
$\tau_{\mathrm{trust}}=0.50$ (strict gate) & $0.960 \pm 0.015$ & $0.940 \pm 0.012$ & $0.886 \pm 0.012$ & $0.908 \pm 0.025$ \\
\midrule
$B=3$ (fewer collaborators) & $0.956 \pm 0.018$ & $0.936 \pm 0.011$ & $0.886 \pm 0.012$ & $0.908 \pm 0.025$ \\
$B=7$ (more collaborators) & $0.959 \pm 0.016$ & $0.939 \pm 0.011$ & $0.886 \pm 0.012$ & $0.908 \pm 0.025$ \\
\bottomrule
\end{tabular}}}
\end{table}

\section{Conclusions and Future Work}
\label{sec:conclusion}

We presented DMTT, a protocol that extends \texttt{MURMURA} to dynamic network topologies under adversarial topology-manipulation attacks. The theoretical analysis shows that DMTT confines Byzantine influence to a bounded residual $\delta_{\max}$ in the mixing matrices, which shrinks as the Beta-evidence trust mechanism accumulates observations. \texttt{MURMURA} was also redesigned into a production-ready distributed framework with coordinator-free wall-clock synchronisation (Eq.~\eqref{eq:wall_clock}), ZeroMQ peer-to-peer transport, a deterministic mobility model (Eqs.~\eqref{eq:random_walk}--\eqref{eq:edge_condition}), and an omission topology-liar attack (Eq.~\eqref{eq:liar_claim}) detectable via an oracle-free self-inclusion check.

Empirically, on UCI HAR and PAMAP2 across $100$ mobile clients under combined topology-liar and Gaussian attacks, DMTT sustained honest-node accuracy above $0.862$ on UCI HAR and $0.829$ on PAMAP2 at every Byzantine fraction from $10\%$ to $80\%$, while static and dynamic FedAvg collapsed to chance. Dedicated robust aggregators (Krum, BALANCE, UBAR) all failed to consistently beat the local-only baseline: Krum degraded at $40$--$60\%$ Byzantine ($0.734$--$0.819$ vs.\ DMTT's $0.880$--$0.910$) as topology-liar poisoning let compromised models pass its distance filter; BALANCE collapsed near-chance on PAMAP2 even at $10\%$ Byzantine ($0.170$); and UBAR collapsed above $50\%$ Byzantine (mean $0.395\pm0.350$ at $50\%$, chance by $70$--$80\%$). DMTT was the only method to clear the local-only bar across both datasets at all tested fractions, degrading gracefully toward near-local-only performance at extreme fractions rather than collapsing, a safety guarantee no other aggregator provides. A trust diagnostic confirmed \emph{exactly zero} Byzantine aggregation weight at every fraction, consistent with the $\delta_{\max}=0$ condition of Theorem~\ref{thm:main}.

Future work includes relaxing the doubly-stochastic assumption on $W^t$ to accommodate asymmetric \texttt{MURMURA}-style weights, extending the analysis to bounded-staleness execution, and validating DMTT on real mobile deployments where neighbours are discovered via physical beacons rather than a shared mobility model, a change that would affect only the routing/discovery step, since the self-inclusion check already relies solely on transport-level delivery acknowledgements.

%\balance
\bibliographystyle{IEEEtran}
\bibliography{references}

@article{rangwala2025evidential,
  title={Evidential Trust-Aware Model Personalization in Decentralized Federated Learning for Wearable IoT},
  author={Rangwala, Murtaza and Sinnott, Richard O and Buyya, Rajkumar},
  journal={arXiv preprint arXiv:2512.19131},
  year={2025}
}

@article{ye2022decentralized,
  title={Decentralized federated learning with unreliable communications},
  author={Ye, Hao and Liang, Le and Li, Geoffrey Ye},
  journal={IEEE journal of selected topics in signal processing},
  volume={16},
  number={3},
  pages={487--500},
  year={2022},
  publisher={IEEE}
}

@article{wu2024topology,
  title={Topology learning for heterogeneous decentralized federated learning over unreliable d2d networks},
  author={Wu, Zheshun and Xu, Zenglin and Zeng, Dun and Li, Junfan and Liu, Jie},
  journal={IEEE Transactions on Vehicular Technology},
  volume={73},
  number={8},
  pages={12201--12206},
  year={2024},
  publisher={IEEE}
}

@article{nesterenko2009discovering,
  title={Discovering network topology in the presence of byzantine faults},
  author={Nesterenko, Mikhail and Tixeuil, S{\'e}bastien},
  journal={IEEE Transactions on Parallel and Distributed Systems},
  volume={20},
  number={12},
  pages={1777--1789},
  year={2009},
  publisher={IEEE}
}

@article{gaucher2024unified,
  title={Unified breakdown analysis for byzantine robust gossip},
  author={Gaucher, Renaud and Dieuleveut, Aymeric and Hendrikx, Hadrien},
  journal={arXiv preprint arXiv:2410.10418},
  year={2024}
}

@article{kharrat2024decentralized,
  title={Decentralized personalized federated learning},
  author={Kharrat, Salma and Canini, Marco and Horvath, Samuel},
  journal={arXiv preprint arXiv:2406.06520},
  year={2024}
}

@inproceedings{liu2024decentralized,
  title={Decentralized directed collaboration for personalized federated learning},
  author={Liu, Yingqi and Shi, Yifan and Li, Qinglun and Wu, Baoyuan and Wang, Xueqian and Shen, Li},
  booktitle={Proceedings of the IEEE/CVF conference on computer vision and pattern recognition},
  pages={23168--23178},
  year={2024}
}

@article{rangwala2025sketchguard,
  title={SketchGuard: Scaling Byzantine-Robust Decentralized Federated Learning via Sketch-Based Screening},
  author={Rangwala, Murtaza and Azzedin, Farag and Sinnott, Richard O and Buyya, Rajkumar},
  journal={arXiv preprint arXiv:2510.07922},
  year={2025}
}

@article{gupta2025fully,
  title={Fully-distributed construction of byzantine-resilient dynamic peer-to-peer networks},
  author={Gupta, Aayush and Pandurangan, Gopal},
  journal={arXiv preprint arXiv:2506.04368},
  year={2025}
}

@article{shi2026dystop,
  title={DySTop: Dynamic Staleness Control and Topology Construction for Asynchronous Decentralized Federated Learning},
  author={Shi, Yizhou and Ma, Qianpiao and Xu, Yan and Zhou, Junlong and Hu, Ming and Liao, Yunming and Xu, Hongli},
  journal={IEEE Transactions on Mobile Computing},
  year={2026},
  publisher={IEEE}
}

@inproceedings{koloskova2020unified,
  title={A unified theory of decentralized SGD with changing topology and local updates},
  author={Koloskova, Anastasia and Loizou, Nicolas and Boreiri, Sadra and Jaggi, Martin and Stich, Sebastian},
  booktitle={International conference on machine learning},
  pages={5381--5393},
  year={2020},
  organization={PMLR}
}

@article{nedic2014distributed,
  title={Distributed optimization over time-varying directed graphs},
  author={Nedi{\'c}, Angelia and Olshevsky, Alex},
  journal={IEEE Transactions on Automatic Control},
  volume={60},
  number={3},
  pages={601--615},
  year={2014},
  publisher={IEEE}
}

@ARTICLE{7405263,
  author={Nedić, Angelia and Olshevsky, Alex},
  journal={IEEE Transactions on Automatic Control},
  title={Stochastic Gradient-Push for Strongly Convex Functions on Time-Varying Directed Graphs},
  year={2016},
  volume={61},
  number={12},
  pages={3936-3947},
  doi={10.1109/TAC.2016.2529285}}

@inproceedings{le2023refined,
  title={Refined convergence and topology learning for decentralized sgd with heterogeneous data},
  author={Le Bars, Batiste and Bellet, Aur{\'e}lien and Tommasi, Marc and Lavoie, Erick and Kermarrec, Anne-Marie},
  booktitle={International Conference on Artificial Intelligence and Statistics},
  pages={1672--1702},
  year={2023},
  organization={PMLR}
}

@article{he2022byzantine,
  title={Byzantine-robust decentralized learning via clippedgossip},
  author={He, Lie and Karimireddy, Sai Praneeth and Jaggi, Martin},
  journal={arXiv preprint arXiv:2202.01545},
  year={2022}
}

@ARTICLE{10208131,
  author={Wu, Zhaoxian and Chen, Tianyi and Ling, Qing},
  journal={IEEE Transactions on Signal Processing},
  title={Byzantine-Resilient Decentralized Stochastic Optimization With Robust Aggregation Rules},
  year={2023},
  volume={71},
  number={},
  pages={3179-3195},
  doi={10.1109/TSP.2023.3300629}}

@inproceedings{anguita2013public,
  title={A public domain dataset for human activity recognition using smartphones},
  author={Anguita, Davide and Ghio, Alessandro and Oneto, Luca and Parra, Xavier and Reyes-Ortiz, Jorge Luis},
  booktitle={European Symposium on Artificial Neural Networks (ESANN)},
  volume={3},
  pages={3},
  year={2013}
}

@article{sensoy2018evidential,
  title={Evidential deep learning to quantify classification uncertainty},
  author={Sensoy, Murat and Kaplan, Lance and Kandemir, Melih},
  journal={Advances in Neural Information Processing Systems},
  volume={31},
  year={2018}
}

@inproceedings{blanchard2017machine,
  title={Machine learning with adversaries: Byzantine tolerant gradient descent},
  author={Blanchard, Peva and El Mhamdi, El Mahdi and Guerraoui, Rachid and Stainer, Julien},
  booktitle={Advances in Neural Information Processing Systems},
  volume={30},
  year={2017}
}

@inproceedings{fang2024byzantine,
  title={Byzantine-robust decentralized federated learning},
  author={Fang, Minghong and Zhang, Zifan and Hairi and Khanduri, Prashant and Liu, Jia and Lu, Songtao and Liu, Yuchen and Gong, Neil},
  booktitle={Proceedings of the 2024 ACM SIGSAC Conference on Computer and Communications Security},
  pages={2874--2888},
  year={2024}
}

@inproceedings{reiss2012introducing,
author = {Reiss, Attila and Stricker, Didier},
title = {Introducing a New Benchmarked Dataset for Activity Monitoring},
year = {2012},
isbn = {9780769546971},
publisher = {IEEE Computer Society},
address = {USA},
doi = {10.1109/ISWC.2012.13},
booktitle = {Proceedings of the 2012 16th Annual International Symposium on Wearable Computers (ISWC)},
pages = {108–109},
numpages = {2},
series = {ISWC '12}
}

@article{guo2021byzantine,
  title={Byzantine-resilient decentralized stochastic gradient descent},
  author={Guo, Shangwei and Zhang, Tianwei and Yu, Han and Xie, Xiaofei and Ma, Lei and Xiang, Tao and Liu, Yang},
  journal={IEEE Transactions on Circuits and Systems for Video Technology},
  volume={32},
  number={6},
  pages={4096--4106},
  year={2021}
}

@ARTICLE{9850408,
  author={Sun, Tao and Li, Dongsheng and Wang, Bao},
  journal={IEEE Transactions on Pattern Analysis and Machine Intelligence}, 
  title={Decentralized Federated Averaging}, 
  year={2023},
  volume={45},
  number={4},
  pages={4289-4301},
  doi={10.1109/TPAMI.2022.3196503}}

@ARTICLE{9713700,
  author={Liu, Wei and Chen, Li and Zhang, Wenyi},
  journal={IEEE Transactions on Signal and Information Processing over Networks}, 
  title={Decentralized Federated Learning: Balancing Communication and Computing Costs}, 
  year={2022},
  volume={8},
  number={},
  pages={131-143},
  doi={10.1109/TSIPN.2022.3151242}}

@InProceedings{pmlr-v108-zantedeschi20a,
  title = 	 {Fully Decentralized Joint Learning of Personalized Models and Collaboration Graphs},
  author =       {Zantedeschi, Valentina and Bellet, Aur\'elien and Tommasi, Marc},
  booktitle = 	 {Proceedings of the Twenty Third International Conference on Artificial Intelligence and Statistics},
  pages = 	 {864--874},
  year = 	 {2020},
  editor = 	 {Chiappa, Silvia and Calandra, Roberto},
  volume = 	 {108},
  series = 	 {Proceedings of Machine Learning Research},
  month = 	 {26--28 Aug},
  publisher =    {PMLR},
  url = 	 {https://proceedings.mlr.press/v108/zantedeschi20a.html}
}

@InProceedings{Li_2022_CVPR,
    author    = {Li, Shuangtong and Zhou, Tianyi and Tian, Xinmei and Tao, Dacheng},
    title     = {Learning To Collaborate in Decentralized Learning of Personalized Models},
    booktitle = {Proceedings of the IEEE/CVF Conference on Computer Vision and Pattern Recognition (CVPR)},
    month     = {June},
    year      = {2022},
    pages     = {9766-9775}
}

@inproceedings{ye2023personalized,
  title={Personalized federated learning with inferred collaboration graphs},
  author={Ye, Rui and Ni, Zhenyang and Wu, Fangzhao and Chen, Siheng and Wang, Yanfeng},
  booktitle={International conference on machine learning},
  pages={39801--39817},
  year={2023},
  organization={PMLR}
}

@article{lu2023privacy,
  title={Privacy-preserving decentralized federated learning over time-varying communication graph},
  author={Lu, Yang and Yu, Zhengxin and Suri, Neeraj},
  journal={ACM Transactions on Privacy and Security},
  volume={26},
  number={3},
  pages={1--39},
  year={2023},
  publisher={ACM New York, NY}
}

@ARTICLE{11271539,
  author={Li, Baosheng and Gao, Weifeng and Deng, Xiumei and Xie, Jin and Xiong, Zehui and Siew, Marie and Guo, Binquan and Mao, Shiwen and Han, Zhu},
  journal={IEEE Transactions on Mobile Computing}, 
  title={Decentralized Federated Learning Over Time-Varying and Heterogeneous Mobile Computing Networks}, 
  year={2026},
  volume={25},
  number={5},
  pages={6688-6704},
  doi={10.1109/TMC.2025.3638598}}

@ARTICLE{9815556,
  author={Fang, Cheng and Yang, Zhixiong and Bajwa, Waheed U.},
  journal={IEEE Transactions on Signal and Information Processing over Networks}, 
  title={BRIDGE: Byzantine-Resilient Decentralized Gradient Descent}, 
  year={2022},
  volume={8},
  number={},
  pages={610-626},
  doi={10.1109/TSIPN.2022.3188456}}

@ARTICLE{9870745,
  author={Che, Chunjiang and Li, Xiaoli and Chen, Chuan and He, Xiaoyu and Zheng, Zibin},
  journal={IEEE Transactions on Parallel and Distributed Systems}, 
  title={A Decentralized Federated Learning Framework via Committee Mechanism With Convergence Guarantee}, 
  year={2022},
  volume={33},
  number={12},
  pages={4783-4800},
  doi={10.1109/TPDS.2022.3202887}}

@article{wu2023byzantine,
  title={Byzantine-resilient decentralized stochastic optimization with robust aggregation rules},
  author={Wu, Zhaoxian and Chen, Tianyi and Ling, Qing},
  journal={IEEE transactions on signal processing},
  volume={71},
  pages={3179--3195},
  year={2023},
  publisher={IEEE}
}

@ARTICLE{10070815,
  author={Tao, Youming and Cui, Sijia and Xu, Wenlu and Yin, Haofei and Yu, Dongxiao and Liang, Weifa and Cheng, Xiuzhen},
  journal={IEEE Transactions on Computers}, 
  title={Byzantine-Resilient Federated Learning at Edge}, 
  year={2023},
  volume={72},
  number={9},
  pages={2600-2614},
  doi={10.1109/TC.2023.3257510}}

@article{perrey2013trail,
  title={TRAIL: Topology authentication in RPL},
  author={Perrey, Heiner and Landsmann, Martin and Ugus, Osman and Schmidt, Thomas C and W{\"a}hlisch, Matthias},
  journal={arXiv preprint arXiv:1312.0984},
  year={2013}
}

@INPROCEEDINGS{7345272,
  author={Li, Chuanyou and Hurfin, Michel and Wang, Yun},
  booktitle={2015 IEEE Trustcom/BigDataSE/ISPA}, 
  title={Reputation Propagation and Updating in Mobile Ad Hoc Networks with Byzantine Failures}, 
  year={2015},
  volume={1},
  number={},
  pages={111-118},
  doi={10.1109/Trustcom.2015.364}}

@ARTICLE{9849010,
  author={Miao, Yinbin and Liu, Ziteng and Li, Hongwei and Choo, Kim-Kwang Raymond and Deng, Robert H.},
  journal={IEEE Transactions on Information Forensics and Security}, 
  title={Privacy-Preserving Byzantine-Robust Federated Learning via Blockchain Systems}, 
  year={2022},
  volume={17},
  number={},
  pages={2848-2861},
  doi={10.1109/TIFS.2022.3196274}}

\vspace{-0.5em}

\begin{IEEEbiographynophoto}{Shubham Vaishnav}
received a Ph.D. from the Department of Computer and Systems Science, Stockholm University, Sweden, in June, 2026. He received Bachelor's and Master's degrees in Computer Science \& Engineering from IIT (ISM), Dhanbad in 2017 and 2019, respectively. His research interests include reinforcement learning, federated  Learning, distributed systems, optimization, and adaptive decision-making in dynamic environments such as IoT.
\end{IEEEbiographynophoto}

\begin{IEEEbiographynophoto}{Murtaza Rangwala}
received his B.Eng. degree (Honours) in Software Engineering from Monash University in 2022. He is currently a Ph.D. candidate with the Quantum Cloud Computing and Distributed Systems (qCLOUDS) Laboratory, School of Computing and Information Systems, University of Melbourne, Australia. His research interests include distributed privacy-preserving machine learning, Byzantine-robust distributed systems, and blockchain-enabled trust frameworks for IoT and cloud computing environments.
\end{IEEEbiographynophoto}

\begin{IEEEbiographynophoto}{Ali Beikmohammadi}
received the B.Sc. degree in electrical engineering from Bu-Ali Sina University, Hamedan, Iran, in 2017, and the M.Sc. degree in electrical engineering from the Amirkabir University of Technology, Tehran, Iran, in 2019. He is 
currently working toward the Ph.D. degree in computer and systems sciences at Stockholm University, Stockholm, Sweden. His research interests include reinforcement learning, deep learning, and federated learning, both in theory and applications.
%received the B.Sc. degree in electrical engineering from Bu-Ali Sina University, Hamedan, Iran, in 2017, and the M.Sc. degree in electrical engineering with the Amirkabir University of Technology, Tehran, Iran, in 2019. Currently, he is working toward the Ph.D. degree in computer and systems sciences with Stockholm University, Sweden, focusing on research areas such as reinforcement learning, deep learning, and federated learning, both in theory and applications.
\end{IEEEbiographynophoto}

\begin{IEEEbiographynophoto}{Sindri Magnússon}
is an Associate Professor in the Department of Computer and Systems Science at Stockholm University, Sweden. He received a PhD in Electrical Engineering from KTH Royal Institute of Technology, Stockholm, Sweden, in 2017. He was a postdoctoral researcher 2018-2019 at Harvard University, Cambridge, MA. His research interests include large-scale distributed/parallel optimization, machine learning, and control.
\end{IEEEbiographynophoto}

\begin{IEEEbiographynophoto}{Rajkumar Buyya}
is a Redmond Barry Distinguished Professor and Director of the Quantum Cloud Computing and Distributed Systems (qCLOUDS) Laboratory at the University of Melbourne, Australia. Recognized as one of the world’s
most highly cited researchers in computer science and software engineering (h-index: 181; g-index:
394; 170,000+ citations), Dr. Buyya is a Fellow of IEEE, a Foreign Fellow of Academia Europaea, and a Fellow of ACM. He co-founded five major IEEE/ACM international conferences—CCGrid, Cluster, Grid, e-Science, and UCC—and served as the Chair of their inaugural meetings. He served as the founding Editor-in-Chief of the IEEE Transactions on Cloud Computing. He is currently serving as Co-Editor-in-Chief of the Journal of Software: Practice and  Experience, which was established 55+ years ago.
\end{IEEEbiographynophoto}

\appendix

This appendix provides: (i) statements and proofs of
the three lemmas (structural invariants, consensus contraction,
descent) supporting the convergence analysis of Section~VI in the
main paper; (ii) full proofs of Theorem~1, Corollary~1, and
Proposition~1 stated in Section~VI; (iii) the statement and proof
of Theorem~2, extending Theorem~1 to $E \ge 1$ local SGD steps
(referenced from but not stated in the main paper); and (iv) a
distributed-backend verification confirming that the coordinator-free
ZeroMQ implementation reproduces the simulation-backend accuracy
trajectories reported in Section~VII of the main paper.

References to assumptions, equations, and statements (e.g.,
Assumption~3, Theorem~1) refer to the main paper. Equation labels
prefixed ``S.'' are local to this document. We use
$J := \tfrac{1}{N}\mathbf{1}\mathbf{1}^\top$ for the averaging
matrix, $U^t := (I-J)\Theta^t$ for the consensus error, $G^t$ for
the stacked stochastic gradients, and $f^\star := \inf_x f(x)$.

\textbf{Notation Note:} Consistent with Section~VI of the main
paper, when analyzing convergence under perfect screening, $N$
denotes the size of the isolated compatible component $\mathcal{C}$,
and $f(x)$ denotes the objective restricted to that component.

\section{Structural properties of $W^t$}
\label{sec:structural}

\begin{lemma}[Structural invariants]
\label{lem:structural}
For every epoch $t$, the matrix $W^t$ produced by Algorithm~1 satisfies:
\begin{enumerate}
    \item \emph{Screened support:} $w_{ij}^t = 0$ for every $j \notin C_i^t \cup \{i\}$.
    \item \emph{Row-stochasticity:} $w_{ij}^t \ge 0$ and $\sum_{j} w_{ij}^t = 1$ for every $i$.
    \item \emph{Double stochasticity (under Assumption~3):} $\sum_i w_{ij}^t = 1$ for every $j$.
\end{enumerate}
The proof is given in the appendix.
\end{lemma}

To state the next results, we stack the local iterates and stochastic gradients into matrices
$$
\Theta^t = \big[\theta_1^t,\dots,\theta_N^t\big]^\top, \quad G^t = \big[g_1(\theta_1^t,\xi_1^t),\dots,g_N(\theta_N^t,\xi_N^t)\big]^\top,
$$
both in $\mathbb{R}^{N \times d}$, and define the \emph{consensus error}
$$
U^t = (I - J)\Theta^t,
$$
which measures how far the local models drift from their mean. With $\hat g^t := \tfrac{1}{N}\sum_{i} g_i(\theta_i^t,\xi_i^t)$, double stochasticity of $W^t$ implies that the average iterate $\bar\theta^t$ evolves according to
\begin{equation}
\label{eq:avg_iterate}
\bar\theta^{t+1} = \bar\theta^t - \eta \hat g^t,
\end{equation}
which is the device that lets the analysis ``see'' the global objective $f$ through the local dynamics.

\subsection{Proof of Lemma 1 (Structural invariants of \texorpdfstring{$W^t$}{W	extasciicircum t})}
\label{sup:lem_structural}

\begin{proof}
Item 1 (screened support) is immediate from the aggregation step of Algorithm~1 (main paper, line~34), which restricts the sum to $j \in C_i^t \cup \{i\}$. Item 2 (row-stochasticity) follows from the normalization $\sum_{j \in C_i^t \cup \{i\}} w_{ij}^t = 1$ used in the aggregation equation of Section~III, together with the nonnegativity of trust-based weights. Item 3 (double stochasticity) is the refinement provided by Assumption~3, which can be enforced by Metropolis--Hastings reweighting over the screened support.
\end{proof}

\section{Consensus contraction lemma}

\begin{lemma}[Consensus contraction]
\label{lem:consensus}
Under Assumptions~1--4, for every $t \ge 0$,
\begin{equation}
\label{eq:consensus}
\mathbb{E}\| U^{t+1} \|_F^2 \le \Big( 1 - \tfrac{p}{2} \Big) \mathbb{E}\| U^t \|_F^2 + \frac{2 \eta^2}{p}\, \mathbb{E}\| (I - J) G^t \|_F^2.
\end{equation}
The proof is given in the appendix.
\end{lemma}

\subsection{Proof of Lemma 2 (Consensus contraction)}
\label{sup:lem_consensus}

\begin{proof}
The local-update and aggregation steps yield $\Theta^{t+1} = W^t (\Theta^t - \eta G^t)$. By double stochasticity (Assumption~3), $W^t J = J W^t = J$ and $J^2 = J$, so $(W^t - J) J = 0$. Subtracting the mean rows from $\Theta^t$ and $G^t$ inside the product is therefore a no-op:
\begin{align*}
U^{t+1} &= (I - J)\, \Theta^{t+1} \\
&= (W^t - J)(\Theta^t - \eta G^t) \\
&= (W^t - J)\big(U^t - \eta (I - J) G^t\big).
\end{align*}
Let $V^t := (I - J) G^t$. Both $U^t$ and $V^t$ have zero column means, so by Assumption~4 (consensus contraction),
$$
\| (W^t - J)(U^t - \eta V^t) \|_F^2 \le (1 - p)\, \| U^t - \eta V^t \|_F^2.
$$
Apply Young's inequality with parameter $\alpha = \tfrac{p}{2(1-p)} > 0$ (the case $p = 1$ is trivial):
$$
\| U^t - \eta V^t \|_F^2 \le (1+\alpha)\, \| U^t \|_F^2 + \Big(1 + \tfrac{1}{\alpha}\Big)\, \eta^2 \| V^t \|_F^2.
$$
With this choice of $\alpha$:
$$
(1-p)(1+\alpha) = 1 - \tfrac{p}{2}, \quad (1-p)\Big(1+\tfrac{1}{\alpha}\Big) = \tfrac{(1-p)(2-p)}{p} \le \tfrac{2}{p}.
$$
Combining and taking expectation yields the bound of Lemma~2.
\end{proof}

\section{Descent lemma on the average iterate}
\label{sec:descent}

\begin{lemma}[Descent inequality]
\label{lem:descent}
Under Assumptions~1--3 and stepsize $\eta \le \tfrac{1}{L}$,
\begin{equation}
\label{eq:descent}
\begin{aligned}
\mathbb{E}[f(\bar\theta^{t+1})] \le\;& \mathbb{E}[f(\bar\theta^t)] - \frac{\eta}{2}\, \mathbb{E}\|\nabla f(\bar\theta^t)\|^2 \\
& + \frac{\eta L^2}{2N}\, \mathbb{E}\|U^t\|_F^2 + \frac{L \eta^2 \sigma^2}{2N}.
\end{aligned}
\end{equation}
The proof is given in the appendix.
\end{lemma}

\subsection{Proof of Lemma 3 (Descent inequality)}
\label{sup:lem_descent}

\begin{proof}
$L$-smoothness of $f$ applied to the average-iterate update $\bar\theta^{t+1} = \bar\theta^t - \eta \hat g^t$ gives
$$
f(\bar\theta^{t+1}) \le f(\bar\theta^t) - \eta \langle \nabla f(\bar\theta^t), \hat g^t \rangle + \frac{L \eta^2}{2} \| \hat g^t \|^2.
$$
Define the local-mean gradient $h^t := \tfrac{1}{N}\sum_i \nabla F_i(\theta_i^t)$, so that $\mathbb{E}[\hat g^t \mid \mathcal{F}^t] = h^t$, where $\mathcal{F}^t$ is the $\sigma$-algebra of all randomness through round $t$. By Assumption~2 and independence of $g_i$ across clients, $\mathrm{Var}(\hat g^t \mid \mathcal{F}^t) \le \sigma^2/N$, so
$$
\mathbb{E}\big[ \| \hat g^t \|^2 \mid \mathcal{F}^t \big] \le \| h^t \|^2 + \tfrac{\sigma^2}{N}.
$$
Taking conditional expectation in the smoothness inequality:
\begin{equation}
\label{eq:S_descent_intermediate}
\begin{aligned}
\mathbb{E}[f(\bar\theta^{t+1}) \mid \mathcal{F}^t] \le\;& f(\bar\theta^t) - \eta \langle \nabla f(\bar\theta^t), h^t \rangle \\
& + \frac{L \eta^2}{2} \| h^t \|^2 + \frac{L \eta^2 \sigma^2}{2N}.
\end{aligned}
\end{equation}
Apply the polarization identity $\langle a, b \rangle = \tfrac{1}{2}(\|a\|^2 + \|b\|^2 - \|a-b\|^2)$:
\begin{align*}
&-\eta \langle \nabla f(\bar\theta^t), h^t \rangle + \tfrac{L\eta^2}{2}\|h^t\|^2 \\
&\quad = -\tfrac{\eta}{2}\|\nabla f(\bar\theta^t)\|^2 + \tfrac{\eta}{2}\|\nabla f(\bar\theta^t) - h^t\|^2 \\
&\quad\quad + \big(\tfrac{L\eta^2}{2} - \tfrac{\eta}{2}\big) \|h^t\|^2.
\end{align*}
For $\eta \le 1/L$, the coefficient of $\|h^t\|^2$ is nonpositive, so that term may be dropped. Jensen's inequality and $L$-smoothness of each $F_i$ bound the bias:
\begin{align*}
\| \nabla f(\bar\theta^t) - h^t \|^2 &= \Big\| \tfrac{1}{N}\sum_i\!\big[\nabla F_i(\bar\theta^t) - \nabla F_i(\theta_i^t)\big] \Big\|^2 \\
&\le \tfrac{L^2}{N}\, \|U^t\|_F^2.
\end{align*}
Substituting back into~\eqref{eq:S_descent_intermediate} and taking total expectation yields Lemma~3.
\end{proof}

\section{Proof of Theorem 1 (Convergence under perfect screening, $E=1$)}
\label{sup:thm_main}

\begin{proof}
We combine Lemmas~2 and~3 through a coupled-inequality argument. Throughout, define
\begin{align*}
a_t &:= \mathbb{E}\|U^t\|_F^2, \quad b_t := \mathbb{E}\|(I-J)G^t\|_F^2, \\
S &:= \sum_{t=0}^{T-1}\mathbb{E}\|\nabla f(\bar\theta^t)\|^2.
\end{align*}

\paragraph{Step 1: Cumulative consensus error}
Lemma~2 gives $a_{t+1} \le (1-\tfrac{p}{2})a_t + \tfrac{2\eta^2}{p} b_t$. With $a_0 = 0$ (consensus initialization), iterating yields $a_t \le \tfrac{2\eta^2}{p}\sum_{s=0}^{t-1}(1-\tfrac{p}{2})^{t-1-s} b_s$. Summing over $t$ and exchanging summation order,
\begin{equation}
\label{eq:S_sum_consensus}
\sum_{t=0}^{T-1} a_t \le \frac{2\eta^2}{p}\sum_{s=0}^{T-1} b_s \sum_{k=0}^\infty\!\Big(1-\tfrac{p}{2}\Big)^{\!k} = \frac{4\eta^2}{p^2}\sum_{t=0}^{T-1} b_t.
\end{equation}

\paragraph{Step 2: Bound on $b_t$}
The variance identity $\sum_i \|g_i^t - \hat g^t\|^2 \le \sum_i \|g_i^t\|^2$ combined with Assumption~2 ($\mathbb{E}\|g_i^t - \nabla F_i(\theta_i^t)\|^2 \le \sigma^2$) gives
$$
b_t \le \sum_{i=1}^N \mathbb{E}\|\nabla F_i(\theta_i^t)\|^2 + N\sigma^2.
$$
Apply $\|a\|^2 \le 2\|a-b\|^2 + 2\|b\|^2$ twice (first with $b = \nabla F_i(\bar\theta^t)$, then with $b = \nabla f(\bar\theta^t)$), use $L$-smoothness, sum over $i$, and apply Assumption~5:
$$
\sum_i \|\nabla F_i(\theta_i^t)\|^2 \le 2L^2 \|U^t\|_F^2 + 4N\tau_{\mathrm{scr}}^2 + 4N\|\nabla f(\bar\theta^t)\|^2.
$$
Taking expectation,
\begin{equation}
\label{eq:S_bt_bound}
\begin{aligned}
b_t \le\;& 2L^2 a_t + 4N\tau_{\mathrm{scr}}^2 + N\sigma^2 \\
& + 4N\,\mathbb{E}\|\nabla f(\bar\theta^t)\|^2.
\end{aligned}
\end{equation}

\paragraph{Step 3: Closing the recursion}
Substituting~\eqref{eq:S_bt_bound} into~\eqref{eq:S_sum_consensus} and collecting the $\sum_t a_t$ term on the left:
$$
\Big(1 - \tfrac{8\eta^2 L^2}{p^2}\Big)\sum_t a_t \le \tfrac{4\eta^2 N}{p^2}\big[\, 4T\tau_{\mathrm{scr}}^2 + T\sigma^2 + 4S \, \big].
$$
For $\eta \le p/(4L)$, the left factor is at least $1/2$, so
\begin{equation}
\label{eq:S_sum_at}
\sum_t a_t \le \frac{32 N \eta^2}{p^2}\big[\, T(\tau_{\mathrm{scr}}^2 + \sigma^2) + S \, \big].
\end{equation}

\paragraph{Step 4: Combining with the descent lemma}
Summing Lemma~3 over $t = 0, \dots, T-1$ and telescoping the $f$-difference:
$$
\tfrac{\eta}{2}\, S \le f(\bar\theta^0) - f^\star + \tfrac{\eta L^2}{2N}\sum_t a_t + \tfrac{L\eta^2 \sigma^2 T}{2N}.
$$
Substituting~\eqref{eq:S_sum_at},
\begin{multline*}
\tfrac{\eta}{2}\, S \le (f^0 - f^\star) + \tfrac{L\eta^2\sigma^2 T}{2N} \\
+ \tfrac{16 L^2 \eta^3}{p^2}\big[\, T(\tau_{\mathrm{scr}}^2 + \sigma^2) + S \, \big].
\end{multline*}
For $\eta \le p/(8L)$, $\tfrac{16L^2\eta^3}{p^2} \le \tfrac{\eta}{4}$, so the $S$-term on the right is absorbed:
\begin{align*}
\tfrac{\eta}{4}\, S \le (f^0 - f^\star) + \tfrac{L\eta^2 \sigma^2 T}{2N} \\
+ \tfrac{16 L^2 \eta^3 T (\tau_{\mathrm{scr}}^2 + \sigma^2)}{p^2}.
\end{align*}
Dividing both sides by $\eta T/4$ yields the bound stated in Theorem~1.
\end{proof}

\section{Extension to local SGD ($E \ge 1$)}
\label{sec:local_sgd}

The convergence result extends to the case where each client performs $E \ge 1$ local SGD steps per round (as in Algorithm~1, line~24), at the cost of one additional assumption: the second moment of the stochastic gradient must be uniformly bounded (Assumption~6). This is a stronger condition than Assumption~2 (bounded \emph{variance}), and it is what controls the local drift $\|\theta_i^{t,k} - \theta_i^t\|$ over the $E$ local steps.

The local update is parameterized by a local stepsize $\eta_l$, with $\theta_i^{t,0} = \theta_i^t$ and
$$
\theta_i^{t,k+1} = \theta_i^{t,k} - \eta_l\, g_i(\theta_i^{t,k}, \xi_i^{t,k}), \quad k=0,\dots,E-1,
$$
followed by aggregation $\theta_i^{t+1} = \sum_j w_{ij}^t \theta_j^{t, E}$. Define the effective per-round stepsize $\eta := E \eta_l$.

\paragraph{Assumption 6: Bounded gradient second moments.}
There exists $G \ge 0$ such that for all $i \in \mathcal{V}$, $t \ge 0$, and points $x$ visited by the iterates,
$$
\mathbb{E}\| g_i(x, \xi_i) \|^2 \le G^2.
$$
This stronger assumption is invoked only in Theorem~\ref{thm:local_sgd} (local SGD with $E > 1$) and Proposition~1 (imperfect screening). It is not needed for Theorem~1. Note that $G^2 \ge \sigma^2 + \sup_x \|\nabla f(x)\|^2$, and the boundedness typically follows from a uniform Lipschitz/bounded-gradient regularization of the loss.

\setcounter{theorem}{1}
\begin{theorem}[Convergence with $E$ local steps, perfect screening]
\label{thm:local_sgd}
Suppose Assumptions~1--6 hold, screening is perfect ($\delta_{\max} = 0$, $\Delta_{\mathrm{scr}}^2 = 0$), and the iterates are initialized at consensus. If the local stepsize satisfies $\eta_l \le \tfrac{1}{2L E}$ (equivalently, $\eta \le \tfrac{1}{2L}$), then for any horizon $T \ge 1$,
\begin{equation}
\label{eq:local_sgd_bound}
\begin{aligned}
\frac{1}{T} \sum_{t=0}^{T-1} \mathbb{E} \big\| \nabla f(\bar\theta^t) \big\|^2
\le\;& \frac{2 \big(f(\bar\theta^0) - f^\star\big)}{\eta T}
+ \frac{2L\eta\, \sigma^2}{NE} \\
& + \frac{(8 + \tfrac{2}{3} p^2)\, L^2\, \eta^2\, G^2}{p^2}.
\end{aligned}
\end{equation}
The proof is given in the appendix.
\end{theorem}

\begin{remark}[Effect of $E > 1$ and scope]
Relative to Theorem~1, two effects are visible in~\eqref{eq:local_sgd_bound}: (i) the noise floor scales as $\sigma^2/(NE)$ rather than $\sigma^2/N$, reflecting the variance reduction from $E$ stochastic gradients per round; and (ii) the heterogeneity-and-drift term is now expressed via $G^2$, which under $G^2 \gtrsim \tau_{\mathrm{scr}}^2 + \sigma^2$ is comparable to (but generally looser than) the corresponding term in Theorem~1. The bounded-gradient Assumption~6 is restrictive (it does not hold for unregularized least-squares or logistic regression on unbounded data); a more refined analysis along the lines of \cite{koloskova2020unified} recovers the same $\mathcal{O}(\sigma/\sqrt{NET} + 1/T)$ rate without bounded gradients, at the cost of more involved constants.
\end{remark}

\subsection{Proof of Theorem 2 (Convergence with \texorpdfstring{$E \ge 1$}{E >= 1} local steps)}
\label{sup:thm_local_sgd}

\begin{proof}
For $i \in \mathcal V$ and $k = 0, \dots, E$, let $\theta_i^{t,k}$ denote the local iterate after $k$ local steps in round $t$, with $\theta_i^{t,0} = \theta_i^t$. Stack the per-step stochastic gradients as $\widetilde G_i^t := \sum_{k=0}^{E-1} g_i(\theta_i^{t,k}, \xi_i^{t,k})$, so the aggregation step reads $\Theta^{t+1} = W^t (\Theta^t - \eta_l \widetilde G^t)$. Let $\hat g^t := \tfrac{1}{NE}\sum_{i,k} g_i(\theta_i^{t,k}, \xi_i^{t,k})$ be the all-sample average. Double stochasticity gives the average-iterate update $\bar\theta^{t+1} = \bar\theta^t - \eta \hat g^t$ with $\eta = E\eta_l$. Define the true local-gradient average as $\bar h^t := \tfrac{1}{NE}\sum_{i,k} \nabla F_i(\theta_i^{t,k})$.

\paragraph{Step 1: Variance of $\hat g^t$}
Decompose the update into the true gradient and the stochastic noise:
$$
\hat g^t = \bar h^t + \underbrace{\tfrac{1}{NE}\sum_{i,k}\!\big[\, g_i(\theta_i^{t,k}, \xi_i^{t,k}) - \nabla F_i(\theta_i^{t,k}) \,\big]}_{=: \hat g_{\mathrm{noise}}^t}.
$$
Let $\mathcal{F}^{t,k}$ be the filtration capturing all randomness up to local step $k$ in round $t$. The noise terms $g_i(\theta_i^{t,k}, \xi_i^{t,k}) - \nabla F_i(\theta_i^{t,k})$ are martingale differences with respect to $\mathcal{F}^{t,k}$ and are independent across clients. Thus, their sum has variance:
$$
\mathbb{E}\| \hat g_{\mathrm{noise}}^t \|^2 \le \frac{NE\,\sigma^2}{(NE)^2} = \frac{\sigma^2}{NE}.
$$
Because the trajectory $\theta_i^{t,k}$ depends on past noise, $\bar h^t$ and $\hat g_{\mathrm{noise}}^t$ are not strictly orthogonal. Applying the inequality $\|a+b\|^2 \le 2\|a\|^2 + 2\|b\|^2$:
\begin{equation}
\label{eq:S_hat_g_var}
\mathbb{E}\| \hat g^t \|^2 \le 2\mathbb{E}\| \bar h^t \|^2 + \frac{2\sigma^2}{NE}.
\end{equation}

\paragraph{Step 2: Local-drift bound (Assumption 6)}
For $k \in \{0,\dots,E-1\}$, the telescoping identity $\theta_i^{t,k} - \theta_i^t = -\eta_l \sum_{\ell=0}^{k-1} g_i(\theta_i^{t,\ell}, \xi_i^{t,\ell})$ combined with Cauchy--Schwarz and Assumption~6 yields
\begin{equation}
\label{eq:S_drift}
\mathbb{E}\|\theta_i^{t,k} - \theta_i^t\|^2 \le \eta_l^2\, k \!\sum_{\ell=0}^{k-1}\! \mathbb{E}\|g_i^{t,\ell}\|^2 \le \eta_l^2\, k^2\, G^2.
\end{equation}

\paragraph{Step 3: Consensus contraction (local version)}
By the same algebra as Lemma~2 applied to $\Theta^{t+1} = W^t(\Theta^t - \eta_l \widetilde G^t)$:
$$
\mathbb{E}\|U^{t+1}\|_F^2 \le \Big(1-\tfrac{p}{2}\Big)\mathbb{E}\|U^t\|_F^2 + \tfrac{2\eta_l^2}{p}\mathbb{E}\|(I-J)\widetilde G^t\|_F^2.
$$
Bounding the second factor using $\|\widetilde G_i^t\|^2 \le E\sum_k \|g_i^{t,k}\|^2$ (Cauchy--Schwarz) and Assumption~6 gives $\mathbb{E}\|(I-J)\widetilde G^t\|_F^2 \le N E^2 G^2$. Using $\eta = E\eta_l$:
$$
\mathbb{E}\|U^{t+1}\|_F^2 \le \Big(1-\tfrac{p}{2}\Big)\mathbb{E}\|U^t\|_F^2 + \tfrac{2 N\eta^2 G^2}{p}.
$$
Iterating from $U^0 = 0$ and summing the geometric series:
\begin{equation}
\label{eq:S_sum_consensus_local}
\sum_{t=0}^{T-1} \mathbb{E}\|U^t\|_F^2 \le \frac{4 T N \eta^2 G^2}{p^2}.
\end{equation}

\paragraph{Step 4: Descent inequality (local version)}
From $L$-smoothness on $\bar\theta^{t+1} = \bar\theta^t - \eta \hat g^t$ and taking total expectation:
$$
\mathbb{E}[f(\bar\theta^{t+1})] \le \mathbb{E}[f(\bar\theta^t)] - \eta\mathbb{E}\langle \nabla f(\bar\theta^t), \hat g^t \rangle + \frac{L\eta^2}{2}\mathbb{E}\|\hat g^t\|^2.
$$
Since $\mathbb{E}[\hat g^t] = \mathbb{E}[\bar h^t]$, we apply~\eqref{eq:S_hat_g_var} and the polarization identity $-\eta\langle a,b\rangle = -\frac{\eta}{2}\|a\|^2 - \frac{\eta}{2}\|b\|^2 + \frac{\eta}{2}\|a-b\|^2$:
\begin{align*}
\mathbb{E}[f(\bar\theta^{t+1})] \le\;& \mathbb{E}[f(\bar\theta^t)] - \eta\mathbb{E}\langle \nabla f(\bar\theta^t), \bar h^t \rangle + L\eta^2\mathbb{E}\|\bar h^t\|^2 \\
&+ \frac{L\eta^2\sigma^2}{NE} \\
=\;& \mathbb{E}[f(\bar\theta^t)] - \frac{\eta}{2}\mathbb{E}\|\nabla f(\bar\theta^t)\|^2 \\
&+ \frac{\eta}{2}\mathbb{E}\|\nabla f(\bar\theta^t) - \bar h^t\|^2  - \Big(\frac{\eta}{2} - L\eta^2\Big)\mathbb{E}\|\bar h^t\|^2 \\
&+ \frac{L\eta^2 \sigma^2}{NE}.
\end{align*}
For $\eta \le \frac{1}{2L}$, the term $(\frac{\eta}{2} - L\eta^2) \ge 0$, allowing us to drop the negative $\|\bar h^t\|^2$ term.
Bounding the bias via Jensen, $L$-smoothness, and $\|\bar\theta^t - \theta_i^{t,k}\|^2 \le 2\|U_i^t\|^2 + 2\|\theta_i^t - \theta_i^{t,k}\|^2$:
\begin{align*}
\mathbb{E}\|\nabla f(\bar\theta^t) - \bar h^t\|^2 &\le \frac{L^2}{NE}\sum_{i,k} \mathbb{E}\|\bar\theta^t - \theta_i^{t,k}\|^2 \\
&\le \frac{2L^2}{N}\mathbb{E}\|U^t\|_F^2 + \frac{2L^2}{NE}\sum_{i,k}\mathbb{E}\|\theta_i^t - \theta_i^{t,k}\|^2.
\end{align*}
Apply~\eqref{eq:S_drift} and $\sum_{k=0}^{E-1} k^2 \le E^3/3$:
$$
\sum_{i,k}\mathbb{E}\|\theta_i^t - \theta_i^{t,k}\|^2 \le \frac{N \eta_l^2 G^2 E^3}{3}.
$$
With $\eta_l^2 E^2 = \eta^2$, this gives $\frac{2L^2}{NE} \cdot \frac{N\eta_l^2 G^2 E^3}{3} = \frac{2L^2 \eta^2 G^2}{3}$. Combining:
\begin{equation}
\label{eq:S_descent_local}
\begin{aligned}
\mathbb{E}[f(\bar\theta^{t+1})] \le\;& \mathbb{E}[f(\bar\theta^t)] - \frac{\eta}{2}\mathbb{E}\|\nabla f(\bar\theta^t)\|^2 \\
& + \frac{\eta L^2}{N}\mathbb{E}\|U^t\|_F^2 + \frac{\eta^3 L^2 G^2}{3} + \frac{L\eta^2 \sigma^2}{NE}.
\end{aligned}
\end{equation}

\paragraph{Step 5: Combine}
Summing~\eqref{eq:S_descent_local} over $t$ and using~\eqref{eq:S_sum_consensus_local}:
\begin{align*}
\frac{\eta}{2}\sum_t \mathbb{E}\|\nabla f(\bar\theta^t)\|^2
&\le (f^0 - f^\star) + \frac{T L \eta^2 \sigma^2}{NE} \\
&\quad + \frac{\eta L^2}{N} \cdot \frac{4 T N \eta^2 G^2}{p^2} + \frac{T \eta^3 L^2 G^2}{3} \\
&= (f^0 - f^\star) + \frac{T L \eta^2 \sigma^2}{NE} \\
&\quad + T \eta^3 L^2 G^2 \Big(\frac{4}{p^2} + \frac{1}{3}\Big).
\end{align*}
Dividing by $\frac{\eta T}{2}$ and noting that $\frac{8}{p^2} + \frac{2}{3} = \frac{8 + 2p^2/3}{p^2}$, we obtain the final bound:
\begin{align*}
\frac{1}{T} \sum_{t=0}^{T-1} \mathbb{E} \big\| \nabla f(\bar\theta^t) \big\|^2
\le &\frac{2 \big(f(\bar\theta^0) - f^\star\big)}{\eta T}
+ \frac{2L\eta\, \sigma^2}{NE} \\
&\quad + \frac{(8 + \tfrac{2}{3} p^2)\, L^2\, \eta^2\, G^2}{p^2}.
\end{align*}
\end{proof}

\section{Proof of Corollary 1 (Screening improves the bound)}
\label{sup:cor_screening}

% \begin{proof}
% The right-hand side of the bound in Theorem~1 is monotone increasing in $\tau^2$ and monotone decreasing in $p$. Substituting the stated inequalities $\tau_{\mathrm{scr}}^2 \le \tau_{\mathrm{raw}}^2$ and $p_{\mathrm{scr}} \ge p_{\mathrm{raw}} - \epsilon$ into Theorem~1 applied with the screened parameters $(\tau_{\mathrm{scr}}^2, p_{\mathrm{scr}})$ yields the claim. When $\epsilon = 0$, the inequality is strict whenever $\tau_{\mathrm{scr}}^2 < \tau_{\mathrm{raw}}^2$.
% \end{proof}

\begin{proof}
The right-hand side of the bound in Theorem~1, when applied to a specific objective and communication graph, is monotone increasing in the data heterogeneity $\tau^2$ and monotone decreasing in the consensus rate $p$. Substituting the stated inequalities $\tau_{\mathrm{scr}}^2 \le \tau_{\mathrm{raw}}^2$ and $p_{\mathrm{scr}} \ge p_{\mathrm{raw}} - \epsilon$ into the heterogeneity-and-topology term of Theorem~1 (applied with the screened component parameters $\tau_{\mathrm{scr}}^2$ and $p_{\mathrm{scr}}$) directly yields the claim. When screening preserves component connectivity ($\epsilon = 0$), the bound strictly improves whenever component isolation successfully reduces data heterogeneity ($\tau_{\mathrm{scr}}^2 < \tau_{\mathrm{raw}}^2$).
\end{proof}

\section{Proof of Proposition 1 (Imperfect screening)}
\label{sup:prop_perturbation}

\begin{proof}
Let $\mathcal{C}$ be the target honest component of size $N$. For $i \in \mathcal{C}$, the actual aggregation step is $\theta_i^{t+1} = \sum_{j\in\mathcal{C}} w_{ij}^t \theta_j^{t+1/2} + \beta_{i,\mathrm{adv}}^t$, where the adversarial injection is $\beta_{i,\mathrm{adv}}^t = \sum_{b\in\mathcal{B}} w_{ib}^t \theta_b^{B, t+1/2}$. 
Define the matrix $B^t \in \mathbb{R}^{N \times d}$ whose $i$-th row is $\beta_{i,\mathrm{adv}}^t$. By Assumption~8, $\|B^t\|_F \le \sqrt{N} \delta_{\max} \Theta_B$.

We rewrite the global update over $\mathcal{C}$ by injecting and subtracting the ideal doubly stochastic matrix $\widetilde W^t$:
$$
\Theta^{t+1} = \widetilde W^t (\Theta^t - \eta G^t) + \underbrace{(W^t - \widetilde W^t)(\Theta^t - \eta G^t) + B^t}_{=: E^t}.
$$

\paragraph{Step 1: Bounding the perturbation $E^t$}
Taking the Frobenius norm and applying the triangle inequality:
$$
\|E^t\|_F \le \|W^t - \widetilde W^t\|_F \big(\|\Theta^t\|_F + \eta \|G^t\|_F\big) + \|B^t\|_F.
$$
By Assumption~9 (bounded domain), $\|\Theta^t\|_F \le \sqrt{N} R$. By Assumption~6, the expected stochastic gradient norm is bounded by $\mathbb{E}[\|G^t\|_F] \le \sqrt{N} G$. Taking the expectation and applying Jensen's inequality to the screening error ($\mathbb{E}\|W^t - \widetilde W^t\|_F \le \Delta_{\mathrm{scr}}$):
\begin{equation}
\label{eq:S_E_bound}
\mathbb{E}\|E^t\|_F \le \sqrt{N} \big[ \Delta_{\mathrm{scr}} (R + \eta G) + \delta_{\max} \Theta_B \big] = \sqrt{N} \rho.
\end{equation}

\paragraph{Step 2: Average-iterate update and Descent Lemma}
Multiplying the update by $\frac{1}{N}\mathbf{1}^\top$, and noting $\mathbf{1}^\top \widetilde W^t = \mathbf{1}^\top$, the average iterate evolves as:
$$
\bar\theta^{t+1} = \bar\theta^t - \eta \hat g^t + e^t,
$$
where $e^t = \frac{1}{N}\mathbf{1}^\top E^t$. By Jensen's inequality, $\mathbb{E}\|e^t\|^2 \le \frac{1}{N} \mathbb{E}\|E^t\|_F^2 \le \rho^2$.
By the $L$-smoothness of $f$:
$$
f(\bar\theta^{t+1}) \le f(\bar\theta^t) + \langle \nabla f(\bar\theta^t), -\eta \hat g^t + e^t \rangle + \frac{L}{2}\|-\eta \hat g^t + e^t\|^2.
$$
Take the conditional expectation $\mathbb{E}[\cdot \mid \mathcal{F}^t]$. Noting $\mathbb{E}[\hat g^t] = h^t$:
\begin{align*}
\mathbb{E}[f^{t+1}] \le\;& f^t - \eta\langle \nabla f(\bar\theta^t), h^t \rangle + \mathbb{E}\langle \nabla f(\bar\theta^t), e^t \rangle \\
& + \tfrac{L}{2}\mathbb{E}\|{-\eta \hat g^t + e^t}\|^2.
\end{align*}
We bound the cross terms using Young's inequality:
\begin{align*}
\langle \nabla f(\bar\theta^t), e^t \rangle &\le \frac{\eta}{8}\|\nabla f(\bar\theta^t)\|^2 + \frac{2}{\eta}\|e^t\|^2 \\
\|-\eta \hat g^t + e^t\|^2 &\le 2\eta^2\|\hat g^t\|^2 + 2\|e^t\|^2.
\end{align*}
Substituting these back yields:
\begin{align*}
\mathbb{E}[f^{t+1}] \le & f^t - \eta\langle \nabla f, h^t \rangle + \frac{\eta}{8}\|\nabla f\|^2 + L\eta^2\mathbb{E}\|\hat g^t\|^2 \\
& + \Big(\frac{2}{\eta} + L\Big)\mathbb{E}\|e^t\|^2.
\end{align*}
Next, we apply the polarization identity $-\eta\langle \nabla f, h^t \rangle = -\frac{\eta}{2}\|\nabla f\|^2 - \frac{\eta}{2}\|h^t\|^2 + \frac{\eta}{2}\|\nabla f - h^t\|^2$ and bound the stochastic variance $\mathbb{E}\|\hat g^t\|^2 \le \|h^t\|^2 + \frac{\sigma^2}{N}$:
\begin{align*}
\mathbb{E}[f^{t+1}] \le\;& f^t - \frac{3\eta}{8}\|\nabla f\|^2 + \frac{\eta}{2}\|\nabla f - h^t\|^2 + \Big(\frac{2}{\eta} + L\Big)\rho^2 \\
&+ \frac{L\eta^2\sigma^2}{N} + \eta\Big(L\eta - \frac{1}{2}\Big)\|h^t\|^2.
\end{align*}
For $\eta \le \frac{1}{2L}$, the factor $(L\eta - 1/2) \le 0$, allowing us to drop the $\|h^t\|^2$ term. Finally, we bound the gradient bias using Jensen's inequality and $L$-smoothness: $\|\nabla f - h^t\|^2 \le \frac{L^2}{N}\|U^t\|_F^2$. Rearranging to isolate the gradient gives:
\begin{align}
\label{eq:S_prop_descent}
\frac{3\eta}{8}\mathbb{E}\|\nabla f(\bar\theta^t)\|^2 \le & f^t - \mathbb{E}[f^{t+1}] + \frac{\eta L^2}{2N}\mathbb{E}\|U^t\|_F^2 + \frac{L\eta^2\sigma^2}{N} \notag \\
&+ \Big(\frac{2}{\eta} + L\Big)\rho^2.
\end{align}

\paragraph{Step 3: Consensus Contraction}
The consensus error $U^{t+1} = (I-J)\Theta^{t+1}$ satisfies:
$$
U^{t+1} = (I-J)\widetilde W^t (U^t - \eta V^t) + (I-J)E^t.
$$
We apply Young's inequality $\|A + B\|_F^2 \le (1+\alpha_1)\|A\|_F^2 + (1+1/\alpha_1)\|B\|_F^2$ with $\alpha_1 = p/2$. Since $\widetilde W^t$ contracts by $(1-p)$, we have $\|A\|_F^2 \le (1-p)\|U^t - \eta V^t\|_F^2$.
Observe that $(1+p/2)(1-p) = 1 - p/2 - p^2/2 \le 1 - p/2$, and $1 + 1/\alpha_1 = 1 + 2/p \le 3/p$. Thus:
$$
\mathbb{E}\|U^{t+1}\|_F^2 \le (1-\tfrac{p}{2})\mathbb{E}\|U^t - \eta V^t\|_F^2 + \frac{3}{p}\mathbb{E}\|E^t\|_F^2.
$$
We apply Young's inequality again to split the inner term: $\|U^t - \eta V^t\|_F^2 \le (1+\alpha_2)\|U^t\|_F^2 + (1+1/\alpha_2)\eta^2\|V^t\|_F^2$, setting $\alpha_2 = p/4$.
Observe that $(1-p/2)(1+p/4) = 1 - p/4 - p^2/8 \le 1 - p/4$, and $(1-p/2)(1+4/p) \le 4/p$. Using~\eqref{eq:S_E_bound} for $\mathbb{E}\|E^t\|_F^2 \le N\rho^2$, we obtain:
\begin{equation}
\label{eq:S_prop_consensus}
\mathbb{E}\|U^{t+1}\|_F^2 \le (1-\tfrac{p}{4})\mathbb{E}\|U^t\|_F^2 + \frac{4\eta^2}{p}\mathbb{E}\|V^t\|_F^2 + \frac{3N}{p}\rho^2.
\end{equation}

\paragraph{Step 4: Combine}
Summing~\eqref{eq:S_prop_consensus} over $t=0$ to $T-1$ bounds the cumulative consensus error:
$$
\sum_{t=0}^{T-1} \mathbb{E}\|U^t\|_F^2 \le \frac{16\eta^2}{p^2}\sum_{t=0}^{T-1} \mathbb{E}\|V^t\|_F^2 + \frac{12 N T}{p^2}\rho^2.
$$
Recall from Step 2 of Theorem 1 that $\mathbb{E}\|V^t\|_F^2 \le 2L^2 \mathbb{E}\|U^t\|_F^2 + 4N\tau_{\mathrm{scr}}^2 + N\sigma^2 + 4N\mathbb{E}\|\nabla f\|^2$. Substituting this into the sum:
\begin{align*}
\sum_{t=0}^{T-1}  \mathbb{E}\|U^t\|_F^2 \le & \frac{32\eta^2 L^2}{p^2}\sum_{t=0}^{T-1}  \mathbb{E}\|U^t\|_F^2 \\
& \quad + \frac{16\eta^2 N}{p^2}\sum_{t=0}^{T-1} \big(4\tau_{\mathrm{scr}}^2 + \sigma^2 + 4\|\nabla f\|^2\big) \\
& \quad + \frac{12 N T}{p^2}\rho^2.
\end{align*}
For $\eta \le \frac{p}{8L}$, we have $\frac{32\eta^2 L^2}{p^2} \le \frac{1}{2}$. Moving the $U^t$ sum to the left and multiplying by 2:
\begin{align*}
\sum_{t=0}^{T-1} \mathbb{E}\|U^t\|_F^2 \le\;& \frac{32\eta^2 N}{p^2}\sum_{t=0}^{T-1}\!\big(4\tau_{\mathrm{scr}}^2 + \sigma^2 + 4\|\nabla f(\bar\theta^t)\|^2\big) \\
& + \frac{24 N T}{p^2}\rho^2.
\end{align*}
We now substitute this consensus sum into the telescoped Descent Lemma~\eqref{eq:S_prop_descent}:
\begin{align*}
\frac{3\eta}{8}\sum_{t=0}^{T-1} \mathbb{E}\|\nabla f\|^2 \le\;& (f^0 - f^\star) + \frac{L\eta^2\sigma^2 T}{N} + \Big(\frac{2}{\eta} + L\Big)\rho^2 T \\
&+ \frac{64\eta^3 L^2 T \tau_{\mathrm{scr}}^2}{p^2} + \frac{16\eta^3 L^2 T \sigma^2}{p^2} +\\
& \frac{12\eta L^2 T}{p^2}\rho^2 + \frac{64\eta^3 L^2}{p^2}\sum_{t=0}^{T-1} \mathbb{E}\|\nabla f\|^2.
\end{align*}
To absorb the gradient term into the left-hand side, we require its coefficient to be $\le \frac{\eta}{8}$, leaving $\frac{\eta}{4}$ on the left. This requires $\frac{64\eta^3 L^2}{p^2} \le \frac{\eta}{8}$, which simplifies to $\eta^2 \le \frac{p^2}{512 L^2}$, yielding the stepsize condition $\eta \le \frac{p}{24L}$.

Assuming this stepsize, subtracting the gradient term leaves $\frac{\eta}{4}\sum\mathbb{E}\|\nabla f\|^2$ on the left. Dividing the entire inequality by $\frac{\eta T}{4}$ yields the final rigorous bound:
\begin{align*}
\frac{1}{T}\sum_{t=0}^{T-1}\mathbb{E}\|\nabla f(\bar\theta^t)\|^2 \le\;& \frac{4(f^0 - f^\star)}{\eta T} + \frac{4L\eta\sigma^2}{N} \\
&+ \frac{256 L^2 \eta^2 (\tau_{\mathrm{scr}}^2 + \sigma^2/4)}{p^2} + \frac{48 L^2}{p^2}\rho^2 \\
&+ \Big(\frac{8}{\eta^2} + \frac{4L}{\eta}\Big)\rho^2.
\end{align*}
\end{proof}

\section{Distributed-Backend Verification}
\label{sup:dist_verify}

To confirm that the simulation backend of Section~V (main paper) faithfully reproduces real distributed execution, we ran DMTT and FedAvg (dynamic) on the ZeroMQ distributed backend at $10$ nodes across three Byzantine fractions ($10\%$, $30\%$, $50\%$) and two seeds each ($20$ rounds, $\Delta = 10\,\mathrm{s}$ per round, IPC transport). Figure~\ref{fig:dist_verify} overlays the per-round honest-node accuracy from both backends.

Across all six (condition, fraction) pairs the two backends produce qualitatively identical trajectories. For DMTT, the mean final honest-node accuracy (last $5$ of $20$ rounds) agrees within $0.07$ between the two backends across all fractions; the largest gap occurs at $30\%$ Byzantine ($0.845$ distributed vs.\ $0.775$ simulation), attributable to the short $20$-round window over which DMTT is still mid-convergence. For FedAvg (dynamic), both backends agree that the unprotected baseline collapses to near-chance ($0.13$--$0.18$) regardless of Byzantine fraction. This qualitative agreement confirms that coordinator-free wall-clock synchronisation and peer-to-peer ZeroMQ transport introduce no systematic deviation from the in-process simulation, justifying the use of the simulation backend for the main $100$-node sweep in Section~VII.

\begin{figure}[H]
\centering
\includegraphics[width=\columnwidth]{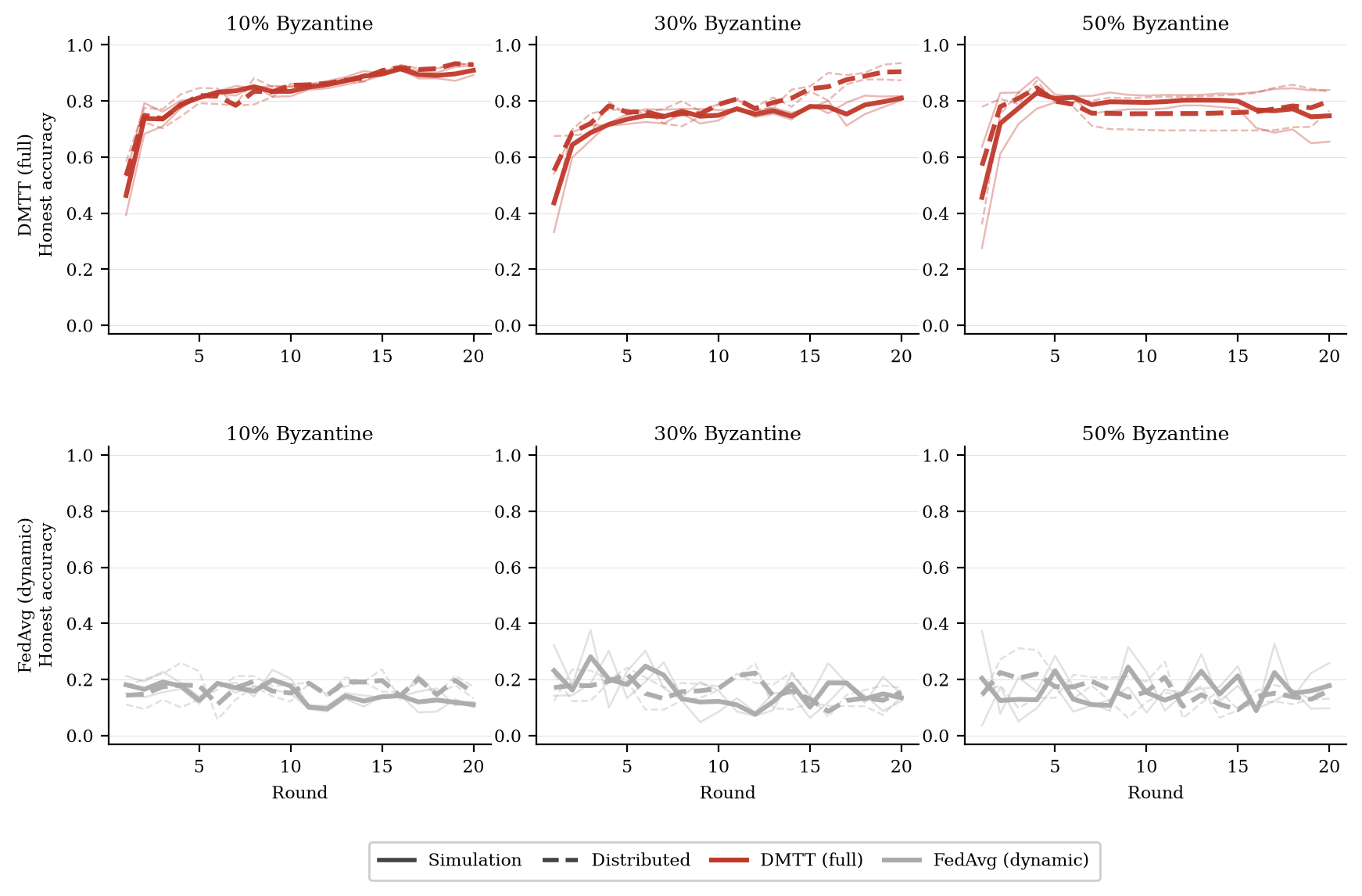}
\caption{Distributed-backend verification at $10$ nodes. Solid lines: simulation backend; dashed lines: ZeroMQ distributed backend (real OS processes, wall-clock synchronisation). Thin lines show individual seeds; bold lines show the mean over $2$ seeds. Rows: DMTT (top) and FedAvg dynamic (bottom). Columns: $10\%$, $30\%$, $50\%$ Byzantine fraction.}
\label{fig:dist_verify}
\end{figure}

\end{document}